\documentclass[journal]{IEEEtran}
\usepackage{cite}
\usepackage{amsmath,amssymb,amsfonts,amsthm}
\usepackage{graphicx}
\usepackage{hyperref}
\usepackage{multirow}
\usepackage{multicol}
\usepackage{booktabs}
\hypersetup{colorlinks=true,urlcolor=red}

\usepackage{outline}
\usepackage{pmgraph}
\usepackage[normalem]{ulem}
\usepackage[utf8]{inputenc}
\usepackage{amssymb}
\usepackage{hyperref}
\usepackage{amsmath}
\usepackage{graphicx}
\usepackage{times}
\usepackage{xcolor}
\usepackage{xspace}
\usepackage[colorinlistoftodos]{todonotes} 
\usepackage{cite}
\usepackage{bm} 
\usepackage{url}

\usepackage{epstopdf}
\AppendGraphicsExtensions{.tiff}
\graphicspath{{fig/}} 

\usepackage{epsfig}
\usepackage{tikz}
\usetikzlibrary{spy}
\usepackage{algpseudocode}
\usepackage{algorithm}
\usepackage{mathrsfs}

\usepackage{nicefrac}       
\usepackage{booktabs}       
\usepackage{bm}  

\usepackage{amsmath,graphicx,caption,mathtools,amssymb}
\usepackage{graphicx}
\usepackage{multirow}
\usepackage{subcaption}
\usepackage{accents}

\usepackage{stmaryrd}

\def\QED{~\rule[-1pt]{5pt}{5pt}\par\medskip}

\long\def\comment#1{} 

\newcommand{\xmath}[1] {\ensuremath{#1}\xspace}
\newcommand{\blmath}[1] {\xmath{\bm{#1}}}

\newcommand{\Wb}{{\blmath W}}

\newcommand{\hb}{{\blmath h}}

\newcommand{\ub}{{\blmath u}}

\newcommand{\xb}{{\blmath x}}
\newcommand{\yb}{{\blmath y}}
\newcommand{\zb}{{\blmath z}}

\newcommand{\Cc}{\mathcal{C}}

\newcommand{\Uc}{\mathcal{U}}

\newcommand{\Sc}{\mathcal{S}}

\newcommand{\Xc}{\mathcal{X}}
\newcommand{\Yc}{\mathcal{Y}}
\newcommand{\Zc}{\mathcal{Z}}

\newcommand{\Phib}{{\boldsymbol {\Phi}}}

\newcommand{\Rd}{{\mathbb R}}

\newcommand{\phib}{{\boldsymbol{\phi}}}

\newcommand{\thetab}{{\boldsymbol {\theta}}}

\newcommand{\beq}{\begin{equation}}
\newcommand{\eeq}{\end{equation}}
\newcommand{\beqa}{\begin{eqnarray}}
\newcommand{\eeqa}{\end{eqnarray}}
\newcommand{\Fc}{{\mathcal F}}

\newtheorem{theorem}{Theorem}
\newtheorem{proposition}[theorem]{Proposition}

\usepackage{pifont}
\usepackage{xcolor}
\newcommand{\cmark}{\textcolor{green!80!black}{\ding{51}}}
\newcommand{\xmark}{\textcolor{red}{\ding{55}}}

\usepackage{dblfloatfix} 

\begin{document}

\title{Cycle-free CycleGAN  using Invertible Generator for Unsupervised Low-Dose CT Denoising}
\date{\vspace{-4ex}}

\author{Taesung~Kwon,
        ~Jong~Chul~Ye,~\IEEEmembership{Fellow,~IEEE}
\thanks{T. Kwon, and J. C. Ye are with the Department of Bio and Brain Engineering, 
		Korea Advanced Institute of Science and Technology (KAIST), 
		Daejeon 34141, Republic of Korea (e-mail: \{star.kwon, jong.ye\}@kaist.ac.kr).
		J.C. Ye is also with the Department of Mathematical Sciences, KAIST.} 
}

\maketitle

\begin{abstract}
Recently, CycleGAN was shown to provide high-performance, ultra-fast denoising 
for low-dose X-ray computed tomography (CT) without the need for a paired training dataset.
Although this was possible thanks to cycle consistency,  CycleGAN requires two generators and two discriminators to enforce cycle consistency, demanding significant GPU resources and technical skills for training.
A recent proposal of tunable CycleGAN with Adaptive Instance Normalization (AdaIN) alleviates the problem in part by using a single generator. However, two discriminators and an additional AdaIN code generator are still required for training.
To solve this problem, here we present a novel {\em cycle-free} Cycle-GAN architecture, which consists of a single generator and a discriminator but still guarantees cycle consistency.
The main innovation comes from the observation that the use of an invertible generator automatically fulfills the cycle consistency condition and eliminates the additional discriminator in the CycleGAN formulation.
To make the invertible generator more effective, our network is implemented in the wavelet residual domain.
Extensive experiments using various levels of low-dose CT images confirm that our method can significantly
improve denoising performance using only 10\% of learnable parameters and faster training time compared to the conventional CycleGAN.
\end{abstract}

\begin{IEEEkeywords}
Low-dose CT, Deep learning, Unsupervised CT denoising, CycleGAN, Invertible Neural Network, Wavelet transform
\end{IEEEkeywords}

\IEEEpeerreviewmaketitle

\begin{figure*}[!t] 	
	\centerline{\includegraphics[width=0.99\linewidth]{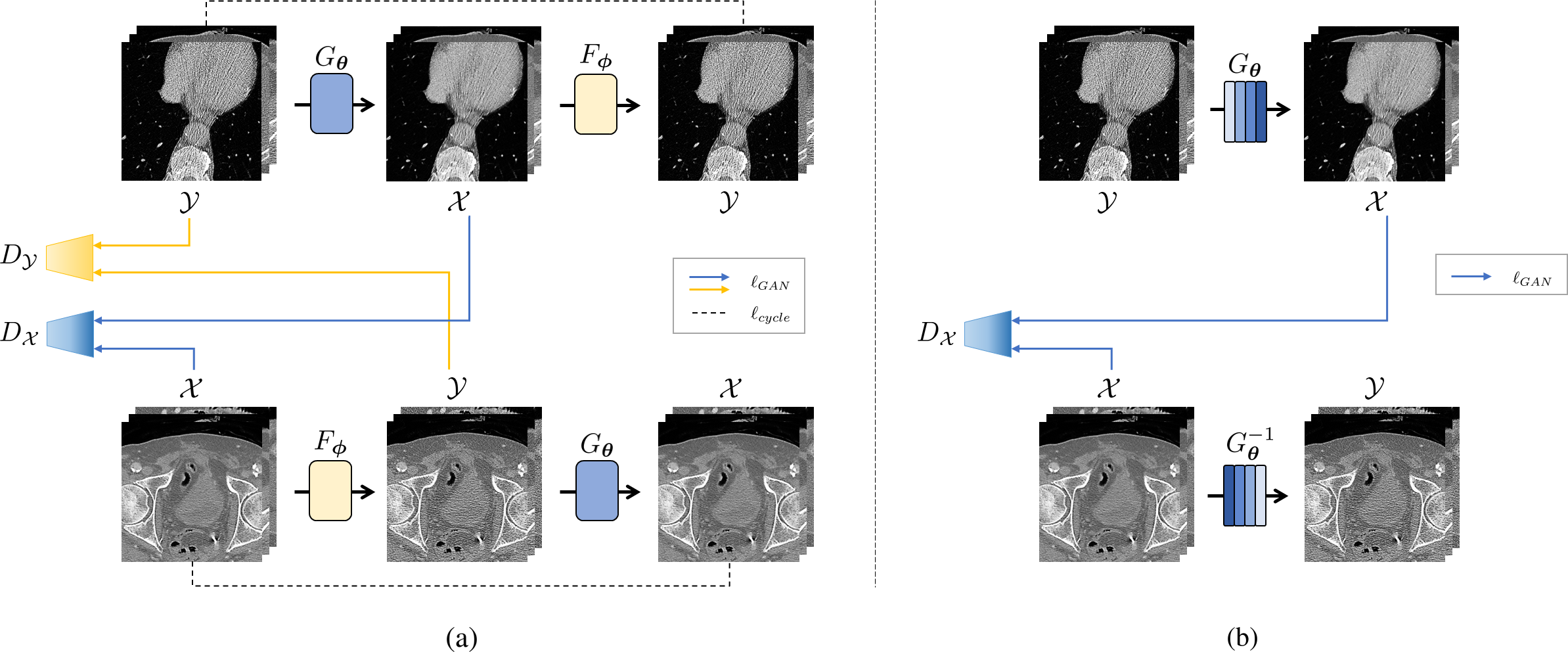}}
	\caption{(a) The architecture of the conventional CycleGAN for CT denoising. 
	Two generators and two discriminators are necessary.
(b) The architecture of our cycle-free CycleGAN with an invertible generator $G_{\theta}$. Only a single pair of a generator and a discriminator is needed.
		}
	\vspace{-0.5cm}
	\label{fig:cyclefreecycleGAN}
\end{figure*}

\section{Introduction}\label{sec:introduction}
\IEEEPARstart{X}{-ray} computed tomography (CT) is one of the most commonly used  medical imaging modalities
with the benefits of high-resolution imaging in a short scan time.
However, excessive X-ray radiation can potentially increase the incidence of cancer, so the
low-dose CT scanning has been extensively studied to minimize the radiation dose to the patients.
Unfortunately,  various artifacts appear on low-dose CT images, which significantly reduces the diagnostic values.

Recently,  deep learning approaches \cite{chen2017low, yang2018low, kang2017deep, kang2018deep, kang2019cycle, gu2021adain} have been
proposed for low-dose CT denoising with impressive performance.  
The majority of these works \cite{chen2017low, kang2017deep, yang2018low, kang2018deep} are based on supervised learning, where
the neural network is trained with paired low-dose CT (LDCT) and standard-dose CT (SDCT) images.
However, the simultaneous acquisition of images with low and high doses is often difficult, which also leads to increased radiation exposure to the subjects.

Accordingly, unsupervised learning approaches that do not require matched LDCT and SDCT images have become a major focus of research in the CT community \cite{kang2019cycle, you2019ct, kim2020unsupervised, gu2021adain}.
In particular, the authors in \cite{kang2019cycle, gu2021adain} proposed a CycleGAN approach \cite{zhu2017unpaired} for low-dose CT denoising that trains the networks with unpaired LDCT and SDCT images.
To enable such unpaired training, 
 two generators are necessary:  one for the forward mapping from LDCT to SDCT, and the other
  for inverse mapping from SDCT  to LDCT.  The cycle consistency is then enforced so that an image that goes through the successive application of forward and inverse mapping should revert to the original one.
   In fact, a recent theoretical study \cite{sim2020optimal} reveals that this CycleGAN architecture emerges as a dual formulation of an optimal
 transport problem where the statistical
 distances between the empirical and transported measures in both source and target domains are simultaneously minimized.

Although two generators are required for training, only the forward generator is used at the time of inference.
Nonetheless, the inverse mapping generator requires a similar number of learnable parameters and memory as the forward mapping, making the CycleGAN architecture inefficient. Furthermore, two generators and two discriminators should be trained simultaneously for convergence, which requires high-level
of skills and know-how for training.
To mitigate this problem,
Gu et al. \cite{gu2021adain}  proposed a tunable CycleGAN with adaptive instance normalization (AdaIN) \cite{huang2017arbitrary}.
The main idea is that a single generator can be switched to the forward or inverse generator
by simply changing the AdaIN code that is 
generated by a lightweight AdaIN code generator.
However, the architecture still requires two discriminators to distinguish the fake and real samples for LDCT and SDCT domains, the individual complexity of which is still as high as that of a generator.

Therefore, one of the ultimate goals of the CycleGAN study for low-dose CT noise removal would be to eliminate the unnecessary generator and discriminator while still maintaining the optimality of CycleGAN from the point of view of optimal transport.
Indeed, one of the most important contributions of this paper is to show that using an invertible generator architecture can automatically satisfy the cycle consistency term and completely remove one of the discriminators without affecting the CycleGAN framework {(see Fig. \ref{fig:cyclefreecycleGAN}).}

To meet the invertibility conditions, our generator is implemented using the coupling layers originally proposed for the normalizing flow \cite{dinh2014nice, dinh2016density}.
Then our generator is trained with just a single discriminator that distinguishes the fake SDCT from the real SDCT images.
To make the invertible generator sufficiently expressive for low-dose CT denoising, our network is trained using the wavelet residual
domain.
 Despite the lack of explicit cycle consistency, our algorithm maintains the optimality of CycleGAN and offers state-of-the-art noise removal with only 10\% of the trainable parameters compared to the conventional CycleGAN.
{Furthermore, the training time is two times faster.}
Since there is no explicit cycle consistency, our method is dubbed as cycle-free CycleGAN.

This paper is structured as follows.
Section~\ref{sec:related works} reviews the existing theory of normalizing flow.
Then,  Section~\ref{sec:theory} explains the mathematical theory behind our cycle-free CycleGAN.
Section \ref{sec:method} explains the implementation issues, training and analysis details, and our low-dose CT datasets. 
Experimental results using the various levels of low-dose CT denoising tasks are shown in Section \ref{sec:result}, which is followed by
 conclusion in Section \ref{sec:conclusion}.

\section{Related works}\label{sec:related works}

\subsection{Normalizing Flow}

Our method is inspired by the normalizing flow (NF) or invertible flow \cite{dinh2014nice, dinh2016density,kingma2018glow}, so we review them
briefly to highlight the similarities and differences from our work. 
However, the original derivation of NF \cite{kingma2013auto,dinh2014nice,rezende2015variational,dinh2016density,kingma2018glow} is difficult to reveal the link to our cycle-free CycleGAN, so here we present
a new derivation, which is inspired from $f$-VAE \cite{su2018f}.

Let $\Xc$ and $\Zc$ denote the ambient space and latent space, respectively.
In classical variational inference,  the model distribution $p_\thetab(\xb),\xb\in \Xc$ is obtained by combining a latent space distribution 
$p(\zb),\zb\in \Zc$ with a  family of conditional distributions $p_\thetab(\xb|\zb)$, which leads to an interesting lower bound:
\begin{align}\label{eq:vi}
\log p_\thetab(\xb) =\log\left(\int p_\thetab(\xb|\zb)p(\zb) d\zb \right) 
\geq -\ell_{ELBO}(\xb;\thetab,\phib) 
\end{align}
\begin{align}
&\ell_{ELBO}(\xb;\thetab,\phib)\notag\\
&:=-\int \log p_\thetab(\xb|\zb) q_\phib(\zb|\xb)d\zb +D_{KL}(q_\phib(\zb|\xb)||p(\zb))   \label{eq:ELBOloss}
\end{align}
where $D_{KL}(q||p)$ denotes the Kullback–Leibler (KL) divergence \cite{kingma2013auto}.
The lower bound in  \eqref{eq:vi}
is often called the evidence lower bound (ELBO) or the variational lower bound \cite{wainwright2008graphical}.
Then, the goal of the
 variational inference tries to find $\thetab$ and the posterior $q_\phib(\zb|\xb)$   that maximize the lower bound.

Among the various choices of posterior $q_\phib(\zb|\xb)$ for the ELBO,  the following form of the posterior is
most often used \cite{su2018f}:
\begin{align}\label{eq:encoder0}
q_\phib(\zb|\xb) = \int  \delta(\zb-F_\phib^\ub(\xb))r(\ub) d\ub
\end{align}
where $r(\ub)$ is zero-mean unit-variance Gaussian,
and $F_\phib^\ub: \xb\in \Xc\mapsto \Zc$ is the encoder function parameterized
by $\phib$ for a given input $\xb\in \Xc$ in addition to noise $\ub$.

For the given encoder in \eqref{eq:encoder0}, the ELBO loss in \eqref{eq:ELBOloss} can be simplified as \cite{su2018f}:
\begin{align}\label{eq:ELBOkey}
& \ell_{ELBO}(\xb;\thetab,\phib) \notag \\ 
:=& -\int  \log  p_\thetab(\xb|F_\phib^\ub(\xb))r(\ub) d\ub  \\
&+\int \log\left(\frac{r(\ub)}{ p(F_\phib^\ub(\xb))} \right)r(\ub) d\ub  \notag \\
&- \int \log \left|\det\left(\frac{\partial F_\phib^\ub(\xb)}{\partial \ub}\right)\right|  r(\ub)d\ub \notag
\end{align}
where the first term  in \eqref{eq:ELBOkey} is obtained from the first term in \eqref{eq:ELBOloss} 
that corresponds to the  likelihood term.
This can be represented as following
by assuming the Gaussian distribution:
\begin{align}\label{eq:encoder}
&-\int  \log  p_\thetab(\xb|F_\phib^\ub(\xb))r(\ub) d\ub \notag \\
= &\int \frac{1}{2} \|\xb- G_\thetab(F_\phib^\ub(\xb))\|^2 r(\ub) d\ub 
\end{align}
where $G_\thetab:\Zc\mapsto \Xc$ is the decoder function parameterized by $\thetab$.
Furthermore, VAE  chooses the following form of the encoder function:
 \begin{align}
 F_\phib^\ub(\xb) = F_\phib(\sigma \ub + \xb)
 \end{align}
 where $\sigma$ is the noise standard deviation,  which is often called the reparametrization trick  \cite{kingma2013auto}.
 Then, the normalizing flow further enforces that
   $F_\phib$ is an invertible function such that
\begin{align}\label{eq:GF}
G_\thetab = F_\phib^{-1} \ .
\end{align}
Thanks to the invertibility condition in \eqref{eq:GF},
 a very interesting phenomenon happens.
More specifically,  \eqref{eq:encoder} can be simplified as follows:
\begin{align}
& \frac{1}{2} \int \| \xb- G_\thetab(F_\phib^\ub(\xb))\|^2 r(\ub) d\ub\notag\\ 
=& \frac{1}{2} \int \| \xb- G_\thetab(F_\phib(\sigma \ub+\xb))\|^2 r(\ub)d\ub \notag\\
=& \frac{1}{2} \int \|\sigma \ub\|^2 r(\ub)d\ub = \frac{\sigma^2}{2} \label{eq:simple}
\end{align} 
which becomes a constant. Therefore, the decoder part is no more necessary 
in the parameter estimation.
Accordingly,  the ELBO loss in \eqref{eq:ELBOkey} can be simplified as
\begin{align}
&\ell_{flow}(\xb,\phib) \label{eq:NF} \\ 
:=&-\int \log\left({ p(F_\phib^\ub(\xb))} \right)r(\ub) d\ub  \notag \\
&- \int \log \left|\det\left(\frac{\partial F_\phib^\ub(\xb)}{\partial \ub}\right)\right|  r(\ub)d\ub\notag
\end{align}
where we have also removed $\int \log r(\ub) r(\ub)d\ub $ term since this is also a constant.
If we further assume the zero mean unit variance Gaussian measure for the latent space $\Zc$,
 \eqref{eq:NF} can be further simplified as
\begin{align}
&\ell_{flow}(\xb,\phib) \label{eq:flow} \\ 
=&\frac{1}{2}\int \| F_\phib(\sigma \ub + \xb)\|^2 r(\ub) d\ub  \notag \\
&- \int \log \left|\det\left(\frac{\partial F_\phib(\sigma \ub + \xb)}{\partial \ub}\right)\right|  r(\ub)d\ub\notag
\end{align}
which is the final loss function for NF.

Now the
main technical difficulty of minimizing the  loss function in \eqref{eq:flow}
arises from the last term which involves with complicated determinant calculation for
huge size matrix. Aside from the invertible network architecture that satisfies
\eqref{eq:GF},  
normalizing flow therefore focuses on encoder function  $F_\phib$  which is composed of a sequence of transformations:
\begin{align}
F_\phib (\ub) = (\hb_K \circ \hb_{K-1} \circ \cdots \circ  \hb_1)(\ub)
\end{align}
For this  encoder function,
the change of variable formula leads to
\begin{align}
 \log \left|\det\left(\frac{\partial F_\phib(\ub)}{\partial \ub}\right)\right| =\sum_{i=1}^K \log \left|\det\left(\frac{\partial \hb_i}{\partial \hb_{i-1}}\right)\right| 
 \end{align}
 where $\hb_0=\ub$.  Note that the complicated determinant computation in   \eqref{eq:flow}
 can be replaced by the relative easy computation for each step  \cite{dinh2014nice}. 

\section{Main Contribution}
\label{sec:theory}

\subsection{Derivation of cycle-free CycleGAN}

\begin{figure}[!t]
	\centering
	\begin{minipage}[b]{0.8\linewidth}
		\centerline{\includegraphics[width=\linewidth]{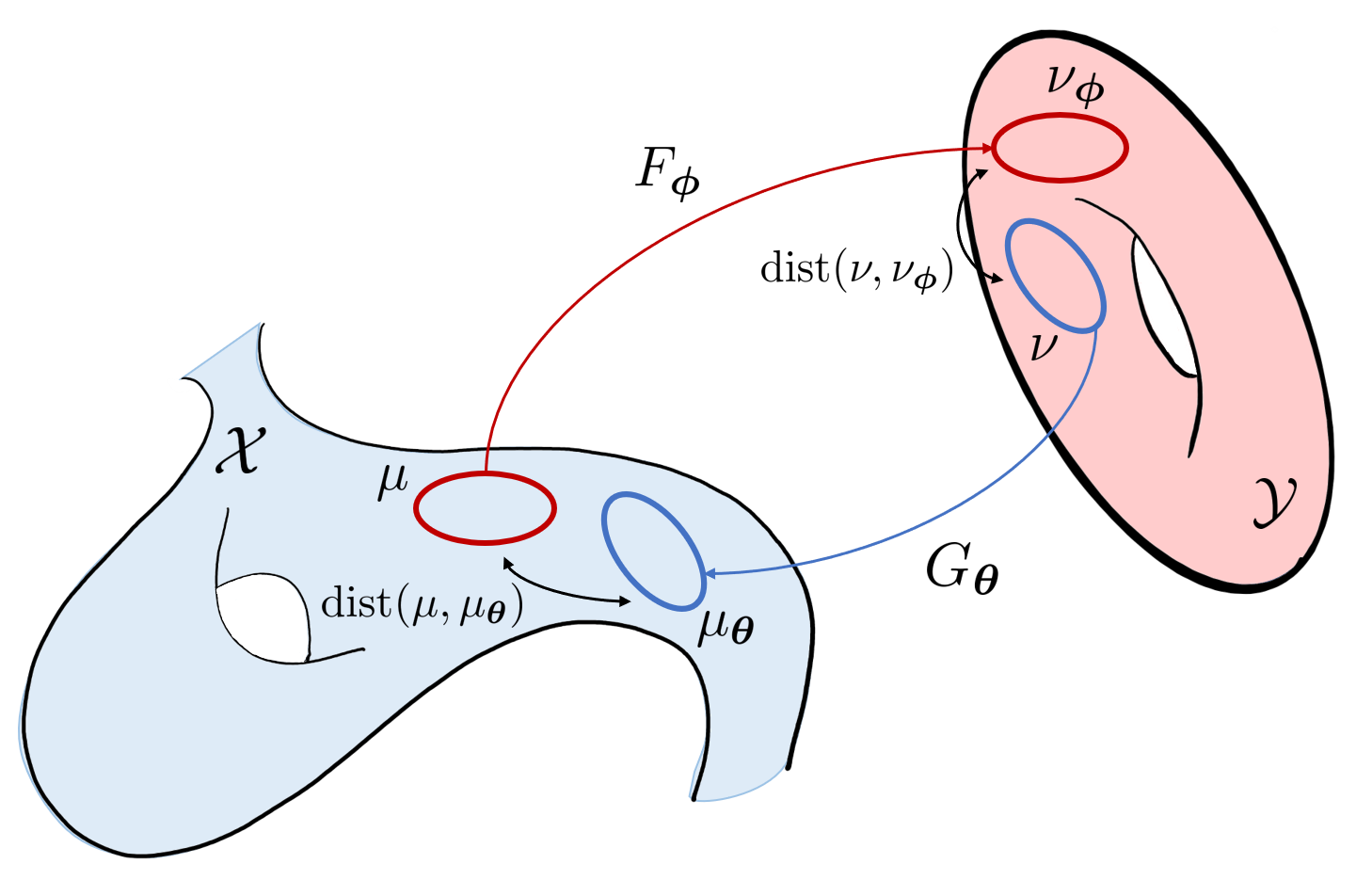}}
		\centerline{(a)}\medskip
	\end{minipage}
	\hspace{0.2cm}
	\begin{minipage}[b]{0.8\linewidth}
		\centerline{\includegraphics[width=\linewidth]{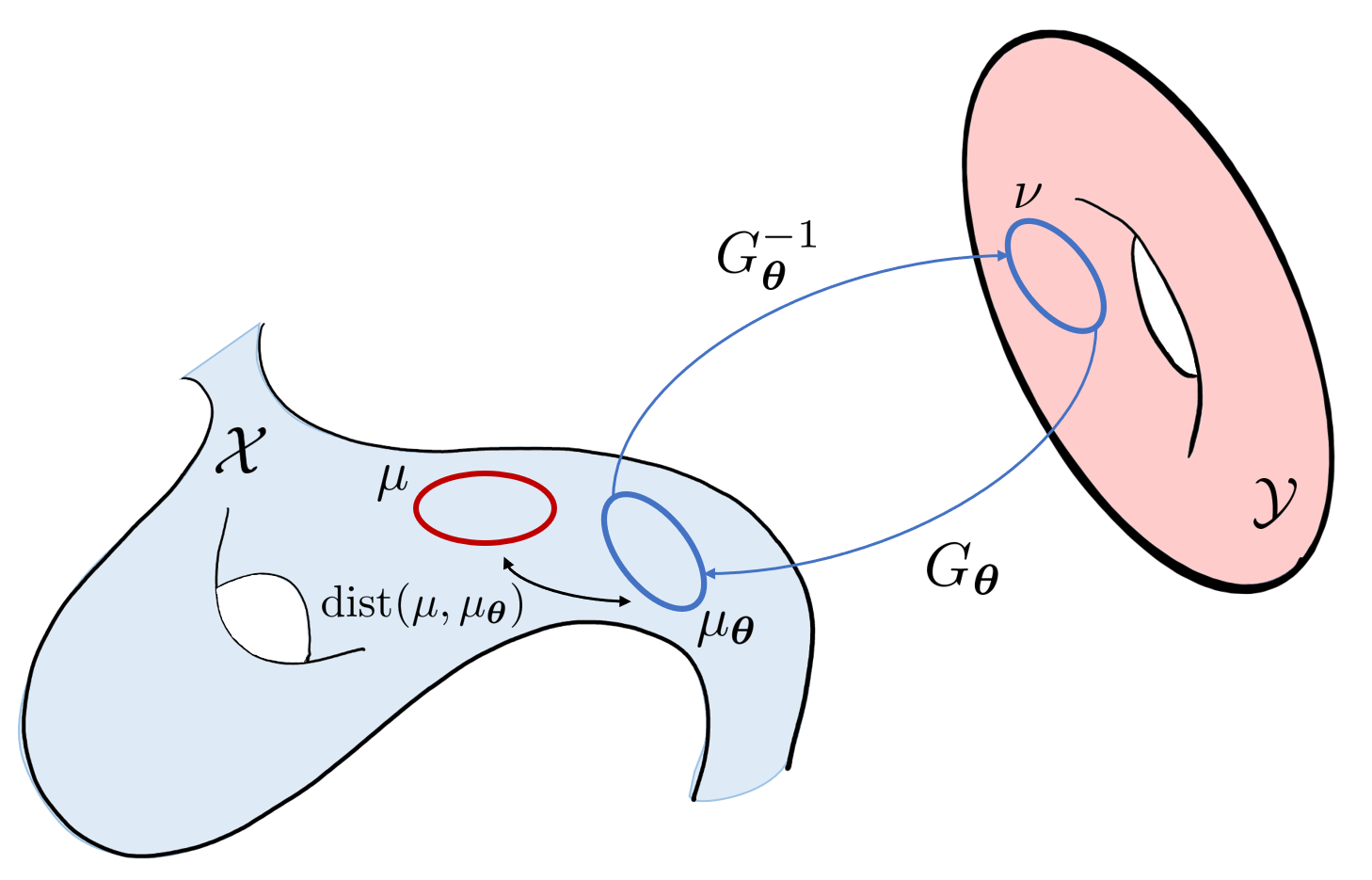}}
		\centerline{(b)}\medskip
	\end{minipage}
	\caption{Geometry of optimal transport: {(a)} Conventional CycleGAN minimizes
	two statistical distances $\mathrm{dist}(\mu,\mu_\thetab)$ and $\mathrm{dist}(\nu,\nu_\phib)$ simultaneously. 
	{(b) Cycle-free CycleGAN minimizes a single distance $\mathrm{dist}(\mu,\mu_\thetab)$ as the other distance can be automatically
	minimized due to the invertibility.}}
	\vspace{-0.5cm}
	\label{fig:optimaltransport}
\end{figure}

Similar to the normalizing flow which is concerned about the conversion between the latent space $\Zc$ and the ambient
space $\Xc$ for image generation, the main goal of CycleGAN is the image transfer
between two spaces, say $\Xc$ and $\Yc$.

Specifically, for the case of low-dose CT denoising, the target SDCT image space 
$\Xc$ is equipped with a probability measure $\mu$, whereas
the  LDCT image space  $\Yc$ is with a probability measure $\nu$ (see  Fig.~\ref{fig:optimaltransport}(a).)
Then, the goal of the CycleGAN is to transport the distribution $\nu$ of the LDCT to the SDCT distribution $\mu$ so that
the LDCT distribution can follow the SDCT distribution. It turns out that this is closely related to the optimal transport
\cite{villani2008optimal,peyre2019computational}.

In particular, the  transport from $(\Xc,\mu)$ to $(\Yc,\nu)$ is performed by the forward operator $F_\phib$,
 so that  $F_\phib$ ``pushes forward'' the measure $\mu$ in $\Xc$ to $\nu_\phib$ in the space $\Yc$ \cite{villani2008optimal,peyre2019computational}.
On the other hand, the mass transportation from the measure space $(\Yc,\nu)$ to another measure space $(\Xc,\mu)$ is done by a generator $G_\thetab: \Yc \mapsto \Xc$, i.e.
the generator $G_\thetab$ pushes forward the measure $\nu$ in $\Yc$ to a measure $\mu_\thetab$ in the target space $\Xc$. 
Then, the optimal transport map for unsupervised learning can be achieved by simultaneously minimizing the statistical distances  $\mathrm{dist}(\mu,\mu_\thetab)$ between $\mu$ and $\mu_\thetab$, and  $\mathrm{dist}(\nu,\nu_\phib)$ 
between $\nu$ and $\nu_\phib$. 

Although various forms of statistical distance could be used (for example, KL divergence in the case of VAE),  
in our prior work \cite{sim2020optimal} and its extension \cite{lim2020cyclegan, oh2020unpaired,cha2021unpaired,chung2021two,lee2020unsupervised}, we use the Wasserstein metric as the statistical distance. Then,  it was shown that the simultaneous
statistical distance minimization can be done by solving the following Kantorovich optimal transport problem:
\begin{align}\label{eq:unsupervised}
\inf\limits_{\pi \in \Pi(\mu,\nu)}\int_{\Xc\times \Yc}c(\xb,\yb;G_\thetab,F_\phib) d\pi(\xb,\yb) 
\end{align}
where $\Pi(\mu,\nu)$ refers to the set of the joint distributions with the margins $\mu$ and $\nu$,
and the transportation cost is defined by
\begin{align}\label{eq:ourc}
c(\xb,\yb;G_\thetab,F_\phib)= \|\xb-G_\thetab(\yb)\|+ \frac{1}{\beta}\|F_\phib(\xb)-\yb\|
\end{align}
where $\beta>0$ denotes some weighting parameter.
In particular, the role of $\beta$ in \eqref{eq:ourc} was originally studied in the context of $\beta$-CycleGAN \cite{lee2020unsupervised}.
In many inverse problems, additional regularization is often used.  For example, one could use the following \cite{cha2021unpaired,chung2021two}:
\begin{align}\label{eq:ourc2}
&c(\xb,\yb;G_\thetab,F_\phib)\notag\\
&= \|\xb-G_\thetab(\yb)\|+ \frac{1}{\beta}\|F_\phib(\xb)-\yb\| +\eta \|\yb-G_\thetab(\yb)\|
\end{align}
where $\eta>0$ is the regularization parameter and the last term penalizes the variation by the generator.
Note that the first two terms in \eqref{eq:ourc2} are computed by using both  $\xb$ and $\yb$,
whereas the last term is only with respect to $\yb$. From the optimal transport perspective, this makes a huge differences,
since the computation of the last term is trivial whereas the first term requires the dual formulation \cite{cha2021unpaired,chung2021two}.

One of the most important contributions of our companion paper \cite{sim2020optimal} is to show that
the primal formulation of the unsupervised learning in \eqref{eq:unsupervised}  with the transport cost \eqref{eq:ourc2}
can be represented by:
\begin{eqnarray}\label{eq:OTcycleGAN}
\min_{\thetab,\phib}\max_{\psi,\varphi}\ell(G_\thetab,F_\phib;\psi,\varphi)
\end{eqnarray}
where 
\begin{align*}
&\ell(G_\thetab,F_\phib;\psi,\varphi)\\
&:=  \lambda \ell_{cycle}(G_\thetab,F_\phib) +\ell_{GAN}(G_\thetab,F_\phib;\psi,\varphi) + \eta \ell_\Yc(G_\thetab)
\end{align*}
where $\lambda>0$ is the hyper-parameter, and  the cycle-consistency term is given by
\begin{align}\label{eq:cycleloss} 
\ell_{cycle}(G_\thetab,F_\phib)  =& \int_{\Xc} \|\xb- G_\thetab(F_\phib(\xb)) \|  d\mu(\xb) \\
&+\frac{1}{\beta}\int_{\Yc} \|\yb-F_\phib(G_\thetab(\yb))\|   d\nu(\yb) \notag 
\end{align}
and 
\begin{align}
&\ell_{GAN}(G_\thetab,F_\phib;\psi,\varphi)  \label{eq:Disc} \\
=&\max_{\varphi\in L^1(\Xc)}\int_\Xc \varphi(\xb)  d\mu(\xb) - \int_\Yc \varphi(G_\thetab(\yb))d\nu(\yb) \notag \\
 & + \max_{\psi\in L^{\frac{1}{\beta}}(\Yc)}\int_{\Yc} \psi(\yb)  d\nu(\yb) - \int_\Xc \psi(F_\phib(\xb))  d\mu(\xb) \notag
\end{align}
where  $L^\kappa(D)$ denotes the space of  $\kappa$-Lipschitz functions with the domain $D$,
and
\begin{align}\label{eq:Yc}
\ell_\Yc(G_\thetab) = \int \|\yb-G_\thetab(\yb)\| d\nu(\yb)
\end{align}
{To make the paper self-contained, see Appendix for the detailed derivation.}

Similar to the key  simplification step \eqref{eq:simple}  in NF,
a very interesting thing happens if we use an invertible generator in \eqref{eq:GF} for the CycleGAN training.
The following proposition is our key result.

\begin{proposition}
Suppose that the generators are invertible, i.e. $G_\thetab=F_\phib^{-1}$ and $F_\phib$ is a $\kappa$-Lipschiz function.
Then, the CycleGAN problem in \eqref{eq:OTcycleGAN} with the transport cost given by \eqref{eq:ourc2}  with $\beta=\kappa$, can be equivalently represented by
\begin{eqnarray}
\min_{\thetab}\max_{\varphi}\ell(G_\thetab;\varphi)
\end{eqnarray}
where 
\begin{eqnarray}
\ell(G_\thetab;\varphi):=  2\ell_{GAN}(G_\thetab;\varphi) + \eta \ell_\Yc(G_\thetab)
\end{eqnarray}
Here, $\ell_\Yc(G_\thetab)$ is defined in \eqref{eq:Yc} and
\begin{align*}
&\ell_{GAN}(G_\thetab;\varphi)  &\\
&=\max_{\varphi\in L^1(\Xc)}\int_\Xc \varphi(\xb)  d\mu(\xb) - \int_\Yc \varphi(G_\thetab(\yb))d\nu(\yb) 
\end{align*}
\end{proposition}

\begin{proof}
First, the invertibility condition  in \eqref{eq:GF} implies that  $F_\phib(G_\thetab(\yb))=\yb$ and $G_\thetab(F_\phib(\xb))=\xb$
 so that  we can easily see that
$\ell_{cycle}(G_\thetab,F_\phib)$ in \eqref{eq:cycleloss} vanishes. 
Second, thanks to the invertibility condition in \eqref{eq:GF}, we have
\begin{align}
 &\max_{\psi\in L^{\frac{1}{\beta}}(\Yc)}\int_{\Yc} \psi(\yb)  d\nu(\yb) - \int_\Xc \psi(F_\phib(\xb))  d\mu(\xb) \notag \\
 &=\max_{\psi\in L^{\frac{1}{\beta}}(\Yc)}\int_{\Xc} \psi(F_\phib(G_\thetab(\yb)))d\nu(\yb) - \int_\Xc \psi(F_\phib(\xb))  d\mu(\xb) \notag \\
&= \max_{\varphi'_\phib\in \Phi} \int_\Yc \varphi'_\phib(G_\thetab(\yb))d\nu(\yb) - \int_\Xc \varphi'_\phib(\xb)  d\mu(\xb) \notag\\
&= \max_{\varphi_\phib\in \Phi}\int_\Xc \varphi_\phib(\xb)  d\mu(\xb) - \int_\Yc \varphi_\phib(G_\thetab(\yb))d\nu(\yb) 
\end{align}
where the set $\Phi$  is defined by
\begin{align}\label{eq:Phi}
\Phi = \{\varphi | \varphi = \psi \circ F_\phib,~\psi \in L^{\frac{1}{\beta}}(\Yc)\}
\end{align}
and the last equality comes from that $\varphi=-\varphi'$, as $-\varphi_\phib\in \Phi$ for $\varphi_\Phib\in \Phi$.
Furthermore, since $F_\phib$ is  a $\kappa$-Lischitz function, we have
\begin{align*}
\|\varphi_\phib(\xb)-\varphi_\phib(\xb')\| & = \|\psi(F_\phib(\xb))-\psi(F_\phib(\xb'))\| \\
&\overset{\mathrm{(a)}}{\leq}\frac{1}{\beta}\|F_\phib(\xb)-F_\phib(\xb')\|\\
&\overset{\mathrm{(b)}}{\leq}\frac{\kappa}{\beta}\|\xb-\xb'\|=\|\xb-\xb'\|
\end{align*}
where the inequality (a) comes from the $1/\beta$-Lipschitz condition of $\psi$ and (b) comes from that $F_\phi$ is a $\kappa$-Lischitz function, where $\kappa=\beta$ by the assumption.
Accordingly, $\varphi_\phib$ is 1-Lipschitz function.
Therefore, we can obtain the following upper bound
\begin{align}\label{eq:upper}
&\max_{\varphi_\phib\in \Phi}\int_\Xc \varphi_\phib(\xb)  d\mu(\xb) - \int_\Yc \varphi_\phib(G_\thetab(\yb))d\nu(\yb) \notag\\
&{\leq} \max_{\varphi\in L^1(\Xc)}\int_\Xc \varphi(\xb)  d\mu(\xb) - \int_\Yc \varphi(G_\thetab(\yb))d\nu(\yb) 
\end{align}
by extending the function space from $\Phi$ to all 1-Lipschitz functions.
Next, we will show that the upper bound in \eqref{eq:upper} is tight.
Suppose that $\varphi^*$ is the maximizer for \eqref{eq:upper}. To show that the bound is tight, we need to show the
existence of $\psi$ such that
\begin{align*}
\varphi^*(\xb)=\psi(F_\phib(\xb)),\quad \forall \xb\in \Xc
\end{align*}
Thanks to the invertibility condition \eqref{eq:GF}, we can always find $\yb\in \Yc$ such that $\xb = G_\thetab(\yb)$ for all $\xb\in \Xc$.
Accordingly,
\begin{align*}
\psi(\yb) =\psi(F_\phib(G_\thetab(\yb)))=\varphi^*(G_\thetab(\yb)),
\end{align*}
which achieves the upper bound. Therefore, we have
\begin{align*}
 &\max_{\psi\in L^{\frac{1}{\beta}}(\Yc)}\int_{\Yc} \psi(\yb)  d\nu(\yb) - \int_\Xc \psi(F_\phib(\xb))  d\mu(\xb) \notag \\
 &=\max_{\varphi\in L^1(\Xc)}\int_\Xc \varphi(\xb)  d\mu(\xb) - \int_\Yc \varphi(G_\thetab(\yb))d\nu(\yb) 
\end{align*}
and $\ell_{GAN}(G,F;\psi,\varphi)=2\ell_{GAN}(G;\varphi)$.
This concludes the proof.
\end{proof}



Compared to the NF, our cycle-free CycleGAN has several advantages.
First, in NF, the latent space $\Zc$ is usually assumed to be Gaussian distribution so that
the main focus is an image generation from noises in $\Zc$ to the ambient space $\Xc$. In order to apply NF to image translation between
$\Yc$ to $\Xc$ domain, we need to implement two NF networks: one for conversion from $\Yc$ to $\Zc$, and the other
from $\Zc$ to $\Xc$. {During this image translation via the latent space, our empirical results shows that the information loss are present due to the restriction to the Gaussian latent variable.}
On the other hand, in our cycle-free CycleGAN, the space $\Xc$ and $\Yc$ can be any empirical distributions.

Additionally, our method has very interesting geometric interpretation.
By replacing the forward operator $F_\phib$ with the inverse of invertible generator $G^{-1}_\thetab$, 
two statistical distance minimization problem in the original CycleGAN in Fig.~\ref{fig:optimaltransport}(a) can be replaced
by the single statistical distance minimization problem as shown in Fig.~\ref{fig:optimaltransport}(b).

\begin{figure}[!bt] 	
	\centerline{\includegraphics[width=\linewidth]{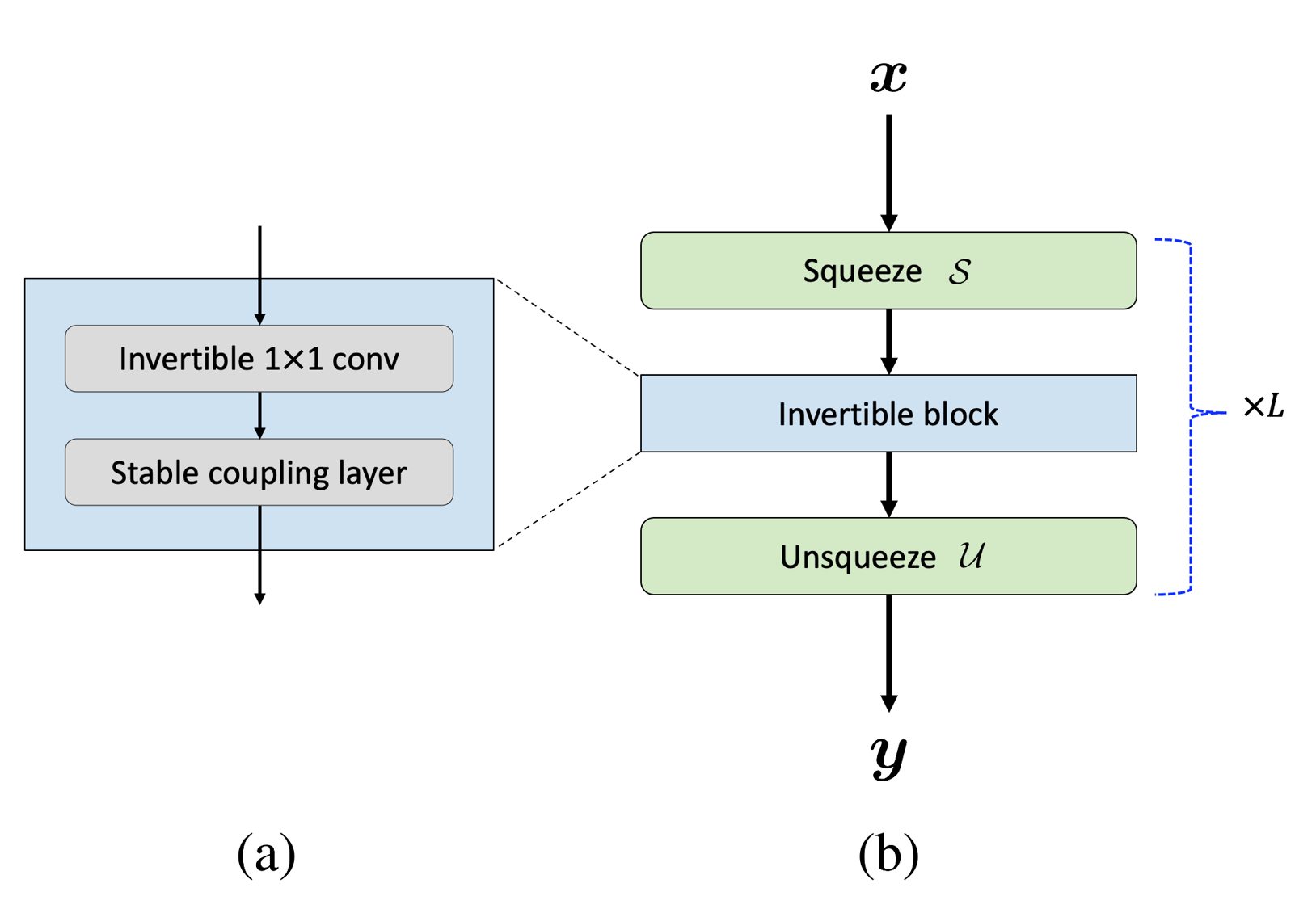}}
	\caption{
    	{The architecture of our (a) Invertible block and (b)  invertible generator $G_{\thetab}$.
        Thanks to its invertibility, we can reverse its operation to get input $\xb$ from output $\yb$ by reversing the operations.}
        }
	\label{fig:flowofgenerator}
\end{figure}

%


\subsection{Invertible Generator}\label{sec:implementation}

%
%
%

Various architecture has been proposed to construct invertible neural networks for flow-based generative models \cite{dinh2014nice, dinh2016density, kingma2018glow}.
For example, the Nonlinear Independent Component Estimation (NICE) \cite{dinh2014nice} is based on
additive coupling layer that leads to volume-preserving invertible mapping.
Later, the method is further extended to the affine coupling layer, which increases the expressiveness of the model \cite{dinh2016density}.

However, this architecture imposes some constraints on the functions that the
network can represent: for instance, it can only represent volume-preserving mappings. Follow-up 
works \cite{dinh2016density,kingma2018glow} addressed this limitation by introducing a new reversible transformation.
More specifically, the authors in  \cite{dinh2016density} proposed a coupling layer
using real-valued non-volume preserving (Real NVP) transformations.
On the other hand, Kingma et al. \cite{kingma2018glow} proposed an invertible 1 $\times$ 1 convolution as a generalization of a permutation operation, which significantly improves the image generation quality of the flow-based generative models.


In the following, we explain specific components of invertible blocks that are used in our method.
Specifically, our network architecture is shown in  Fig. \ref{fig:flowofgenerator}, {which is composed of $L$ repetition of
squeeze/unsqueeze block interleaved with invertible 1$\times$1 convolution and stable additive coupling layers.}
The detailed explanation follows.

\subsubsection{Squeeze and Unsqueeze operation} 

{Squeeze} operation $\Sc$ splits input image $\xb$ into four sub-images   which are arranged along
channel direction as shown in Fig. \ref{fig:squeezeunsqueeze}.
Mathematically, this can be written by
\begin{align*}
\xb_{1:4}:=[\xb_1,\xb_2,\xb_3,\xb_4] =\Sc(\xb)
\end{align*}
Squeeze operation is essential to build the coupling layer, which becomes evident soon. 
{Unsqueeze} operation $\Uc$, denoted by
\begin{align*}
\xb = \Uc(\xb_{1:4}) \ ,
\end{align*}
 then rearranges separated channels  into one image  
as an inverse operation of squeeze operation (see Fig. \ref{fig:squeezeunsqueeze}). 
 This operation is applied using the output of the coupling layer, so that unsqueezed output  maintains the same spatial dimension of the input image $\xb$.

\begin{figure}[!h]
	\centerline{\includegraphics[width=0.9\linewidth]{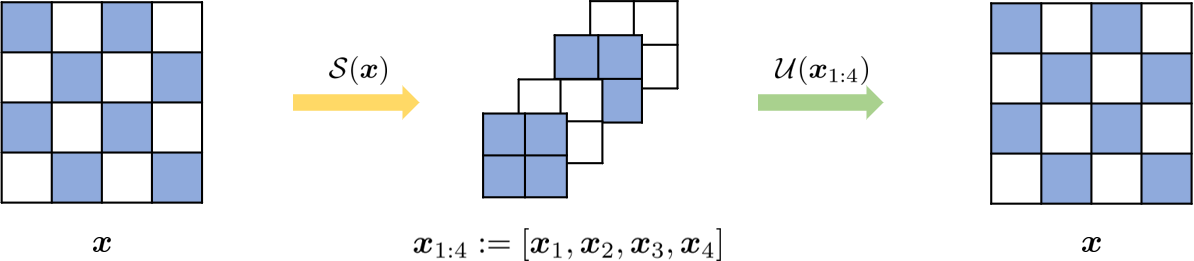}}
	\caption{Squeeze $\Sc$ and unsqueeze operation $\Uc$.
	}
	\label{fig:squeezeunsqueeze}
\end{figure}

\subsubsection{Invertible 1$\times$1 convolution}
The squeeze operation decomposes the input into four components along the channel dimension.
With the resulting fixed channel arrangement, only limited spatial information passes through the neural network. 
Therefore, random shuffling and reversing the order of channel dimension \cite{dinh2014nice, dinh2016density} were proposed. 
On the other hand, Generative Flow with Invertible 1$\times$1 Convolutions (Glow) \cite{kingma2018glow} proposed an
invertible 1$\times$1 convolution with an equal number of input and output channels as a generalization of permutation operation with learnable parameters. 

Mathematically, 1$\times$1 convolution $\Cc$ can be represented by multiplying a matrix $\Wb\in \Rd^{4\times 4}$ as follows:
\begin{align}\label{eq:Cc}
\Cc \left(\xb_{1:4}\right) =  \xb_{1:4}\Wb
\end{align}
which is illustrated in Fig. \ref{fig:1by1conv}.
By multiplying a fully populated matrix, the channel-wise separated input information can be mixed together so that
the subsequent operation can be applied more efficiently.
Then, the corresponding inversion operation $\Cc^{-1}$ can be written by \cite{kingma2018glow}
\begin{align}\label{eq:invconv}
\Cc^{-1} \left(\yb_{1:4}\right) =  \yb_{1:4}\Wb^{-1}
\end{align}
if $\Wb$ is invertible.


\begin{figure}[!ht]
	\centerline{\includegraphics[width=0.9\linewidth]{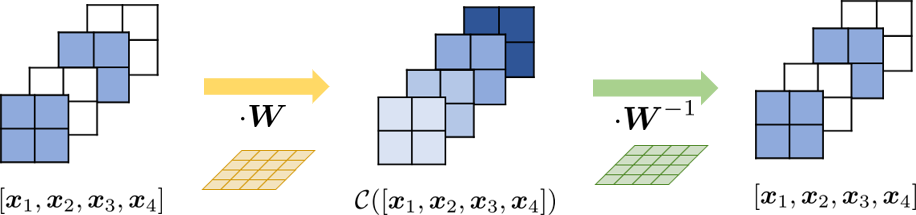}}
	\caption{Invertible 1$\times$1 convolution and its inverse.
	}
	\label{fig:1by1conv}
\end{figure}

\subsubsection{Stable Additive coupling layer}

 {Coupling layer} is the essential component that gives invertibility  but also provides expressiveness
 of the neural network.
The additive coupling layer in NICE \cite{dinh2014nice} is based on the even and odd decomposition
of the sequence, after which neural networks are applied in an alternating manner.

Tomczak et al \cite{tomczak2020general} further extended the additive
coupling layer to general coupling layer where input image can split into four channel
blocks, and neural networks are applied at every step.
%
%
%
By applying the general invertible transformation, we can handle separated input more efficiently.

%

Specifically, the stable coupling layer is given by
\begin{align}\label{eq:forward}
\yb_1 &= \xb_1+\Fc_1\left([\xb_2,\xb_3,\xb_4]\right) \notag\\
\yb_2 &= \xb_2+\Fc_2\left([\yb_1,\xb_3,\xb_4]\right) \notag\\
\yb_3 &= \xb_3+\Fc_3\left([\yb_1,\yb_2,\xb_4]\right) \notag\\
\yb_4 &= \xb_4+\Fc_4\left([\yb_1,\yb_2,\yb_3]\right) 
\end{align}
where $\Fc_i(\cdot), i=1,\cdots, 4$ are neural networks.
Then, the block inversion can be readily done by
\begin{align}
\xb_4 &= \yb_4-\Fc_4\left([\yb_1,\yb_2,\yb_3]\right) \notag\\
\xb_3 &= \yb_3-\Fc_3\left([\yb_1,\yb_2,\xb_4]\right) \notag\\
\xb_2 &= \yb_2-\Fc_2\left([\yb_1,\xb_3,\xb_4]\right) \notag\\
\xb_1 &= \yb_1-\Fc_1\left([\xb_2,\xb_3,\xb_4]\right)
\end{align}
For example, additive operation $\yb_1 = \xb_1+\Fc_1\left([\xb_2,\xb_3,\xb_4]\right)$ and its inverse operation $\xb_1 = \yb_1-\Fc_1\left([\xb_2,\xb_3,\xb_4]\right)$ are shown as Fig. \ref{fig:generaladditive}.

\begin{figure}[!ht]
	\centering
	\begin{minipage}[b]{0.8\linewidth}
		\centerline{\includegraphics[width=\linewidth]{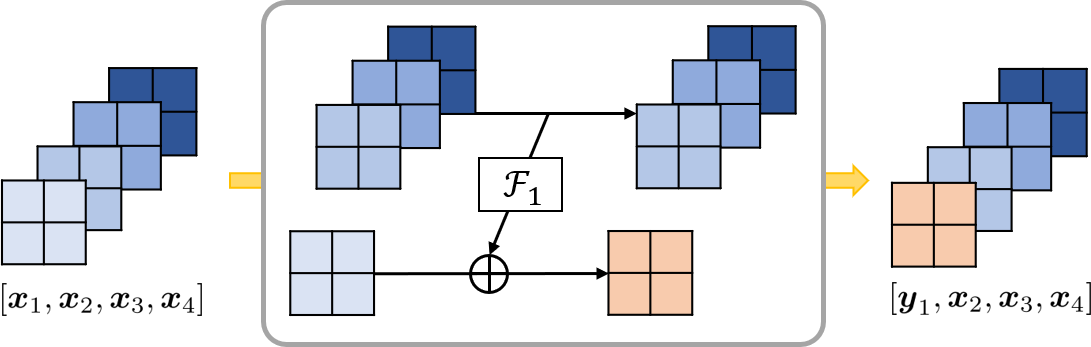}}
		\centerline{(a)}\medskip
	\end{minipage}
	\begin{minipage}[b]{0.8\linewidth}
		\centerline{\includegraphics[width=\linewidth]{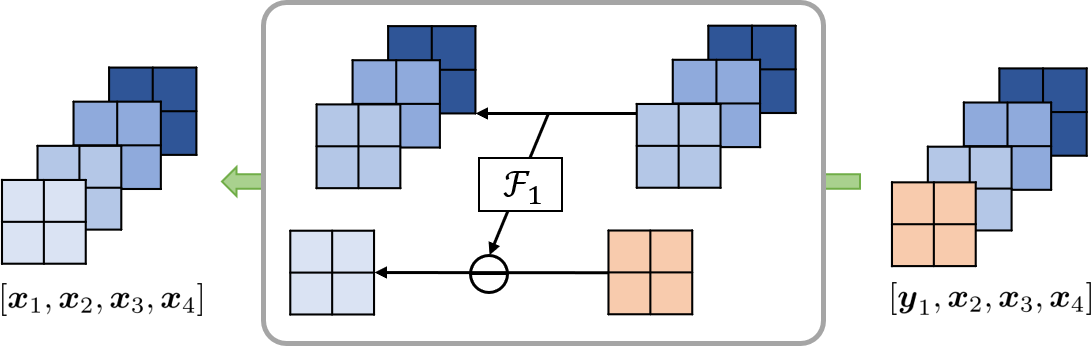}}
		\centerline{(b)}\medskip
	\end{minipage}
	\caption{Implementation of (a)  single additive operation $\yb_1 = \xb_1+\Fc_1\left([\xb_2,\xb_3,\xb_4]\right)$. and (b)  its inverse operation $\xb_1 = \yb_1-\Fc_1\left([\xb_2,\xb_3,\xb_4]\right)$.}
	\label{fig:generaladditive}
\end{figure}

\subsubsection{Lipschitz constant computation}

It is easy to see that the Jacobian of the stable coupling layer has unit determinant \cite{dinh2014nice}.
In fact, among the aforementioned modules in the invertible networks,
only module that does not have unit determinant is the 1$\times$1 convolution layer. Specifically,
  the log-determinant of the step \eqref{eq:Cc}  is determined by that of $\Wb$ \cite{kingma2018glow}:
\begin{align}
 \log \left|\det\left(\frac{d \Cc(\xb_{1:4})}{d \xb_{1:4}}\right)\right| = 
 \log \left| \det(\textbf{W}) \right| 
 \end{align}
 Similarly, the Lipschitz constant for the invertible generator can be easily determined by the matrix norm of $\Wb$.

%

\subsection{Wavelet Residual Learning}


Unlike the image generation from noises,
one of the important observations in the image denoising is that
the noisy and clean images share structural similarities.
Accordingly, rather than learning all components of the images, the authors in \cite{song2020unsupervised,gu2021adain}
proposed wavelet residual domain learning approach, and we follow the same procedure.

Specifically, as shown in Fig. \ref{fig:waveletresidual}(a), wavelet decomposition separates high-frequency component and low-frequency components,
then by nulling only the low-frequency (LL) component at the last level decomposition, we can obtain
the wavelet residual images that contain high-frequency components.
Then, as shown in Fig. \ref{fig:waveletresidual}(b), our network is trained using only high-frequency components.
This makes the network handles CT noise components much easier, because most of the CT noises are concentrated in high frequency and
the common low-pass images are not processed by neural networks.



\begin{figure}[!ht]
	\centering
    \begin{minipage}[b]{\linewidth}
		\centerline{\includegraphics[width=\linewidth]{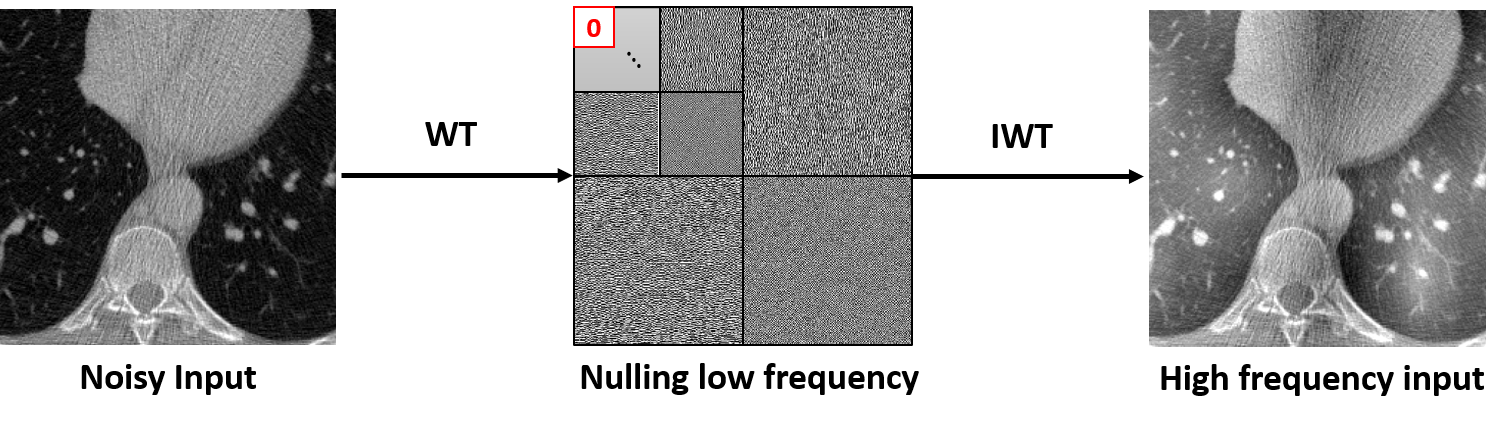}}
		\centerline{(a)}\medskip
	\end{minipage}
	\hspace{0.2cm}
	\begin{minipage}[b]{\linewidth}
		\centerline{\includegraphics[width=\linewidth]{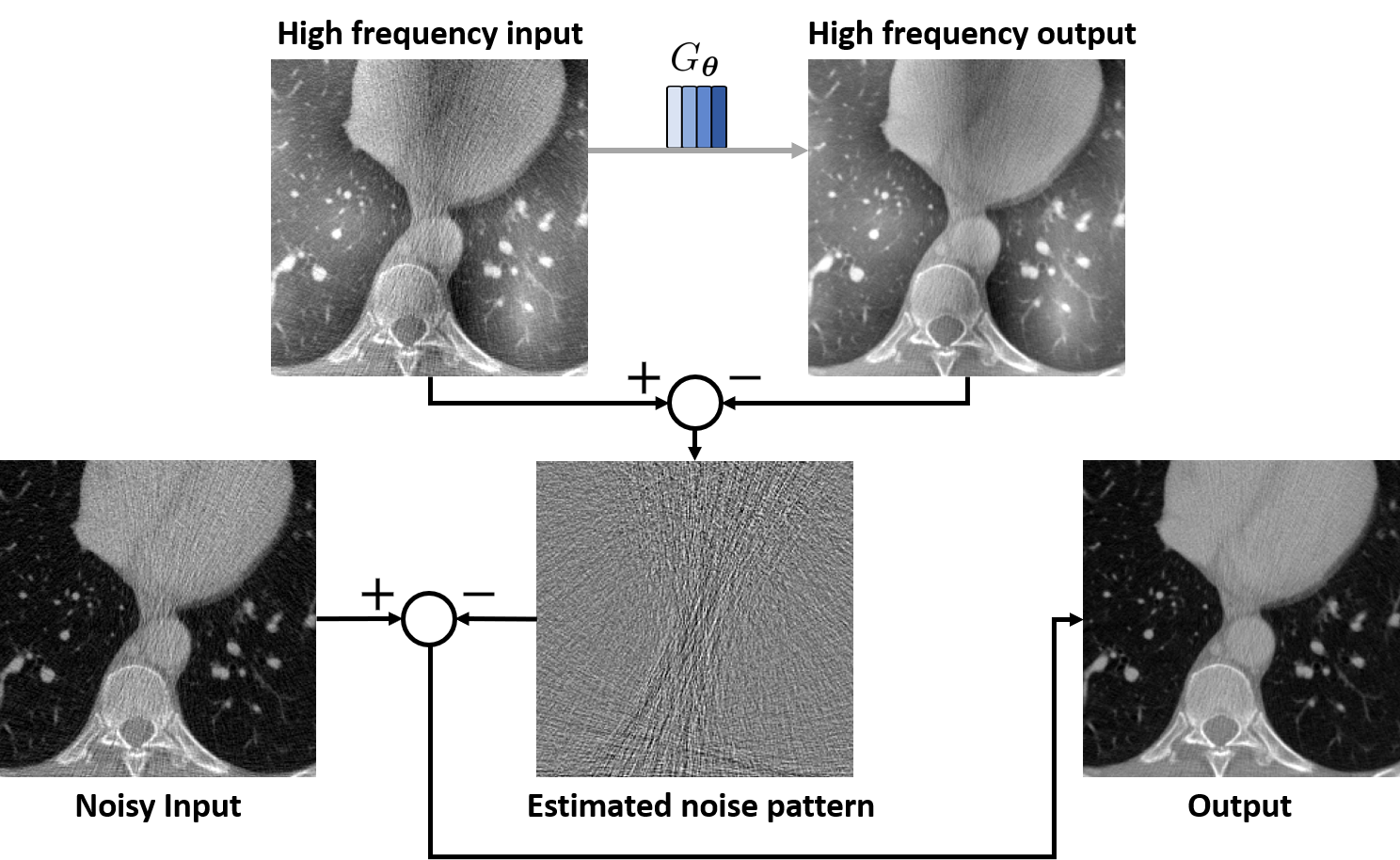}}
		\centerline{(b)}\medskip
	\end{minipage}
	\caption{(a) Generation of  wavelet residual image by nulling out the lowest band of the wavelet decomposition. 
	(b) Our network is trained using wavelet residual. Noise pattern is acquired, and then
	the final image is obtained by subtracting the noise pattern.
	}
	\label{fig:waveletresidual}
\end{figure}

\section{Method}\label{sec:method}
\subsection{Dataset}

To verify the denoising performance of our framework, we use two datasets, one for the quantitative analysis and the other  for the qualitative analysis. For quantitative analysis, we use paired low-dose and standard-dose CT image dataset which was used for study by Kang et al.\cite{kang2017deep}.
Specifically, the data are abdominal CT projection data from the AAPM 2016 Low Dose CT Grand Challenge.
For qualitative experiments, we use unpaired 20\% dose cardiac multiphase CT scan dataset which was used for study by Kang et al.\cite{kang2019cycle}.
The details are as follows.

\subsubsection{AAPM CT dataset} 
AAPM CT dataset is a reconstructed CT image dataset from abdominal CT projection data in the AAPM 2016 Low Dose CT Grand Challenge, which was used for study by Kang et al.\cite{kang2017deep}. Total 10 patients' data were obtained after approval by the institutional review board of the Mayo Clinic. CT images with the size of $512\times512$ were reconstructed using a conventional filtered backprojection algorithm. Poisson noise was inserted into the projection data to make noise level corresponded to 25\% of the standard-dose. As the low-dose CT image data were simulated based on standard-dose CT images, they are paired dataset.
For the training, every value of the dataset is converted into Hounsfield unit [HU] and the value lower than -1000 HU is truncated to -1000 HU.
Then, we divide the dataset into 4000 to normalize all data values between [-1, 1]. 
To train our network, we use 3839 CT images, while the other 421 images were used to test our network.

\subsubsection{20\% dose Multiphase Cardiac CT scan dataset}
The 20\% dose cardiac multiphase CT scan dataset was acquired from 50 CT scans of mitral value prolapse patients and 50 CT scans of coronary artery disease patients. {The dataset was collected at the University
of Ulsan College of Medicine and used for study by Kang et al.\cite{kang2019cycle} and Gu et al.\cite{gu2021adain}.}
The detailed information of CT scan protocol is described in previous reports \cite{koo2014demonstration, yang2015stress}.
Electrocardiography (ECG)-gated cardiac CT scanning with second-generation dual-source CT scanner was performed.
For the low-dose CT scan, the tube current is reduced to 20\% of the standard-dose CT scan.
For the training, every value of the dataset is converted into Hounsfield unit [HU] and the value lower than -1024 HU is truncated to -1024 HU.
After that, we divide the dataset into 4096 to normalize all data values between [-1, 1]. 
To train our network, we use 4684 CT images, while the other 772 images were used to test our model. 


\begin{figure}[!hbt]
    \centering
    \begin{minipage}[b]{0.48\linewidth}
		\centerline{\includegraphics[width=\linewidth]{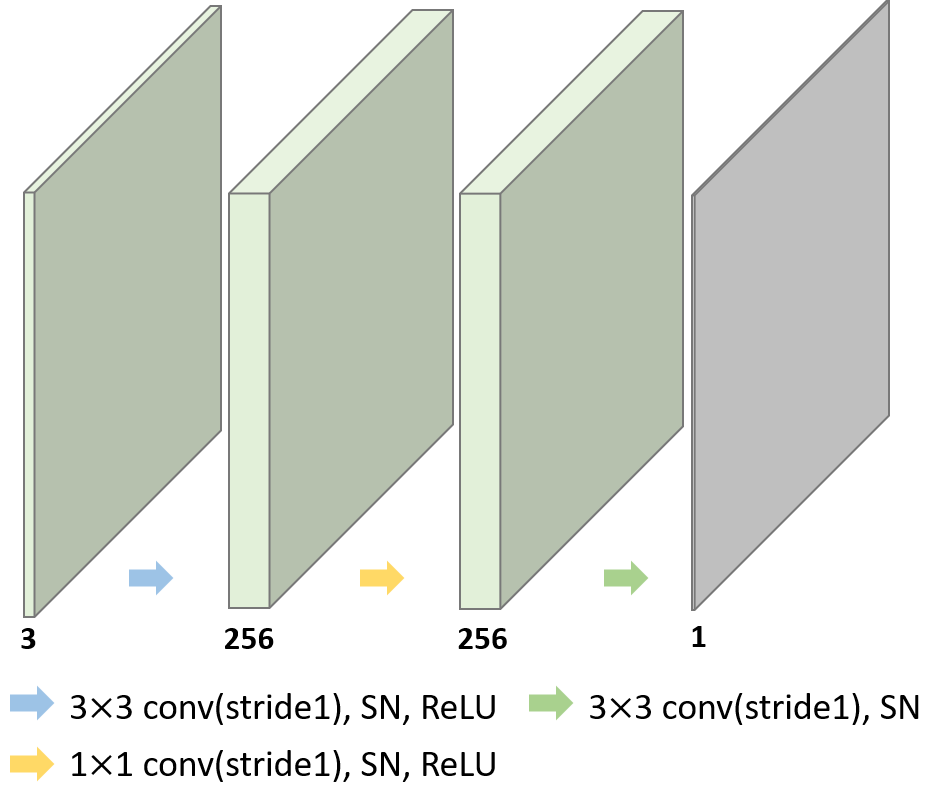}}
		\centerline{(a)}\medskip
	\end{minipage}
	\begin{minipage}[b]{0.5\linewidth}
		\centerline{\includegraphics[width=\linewidth]{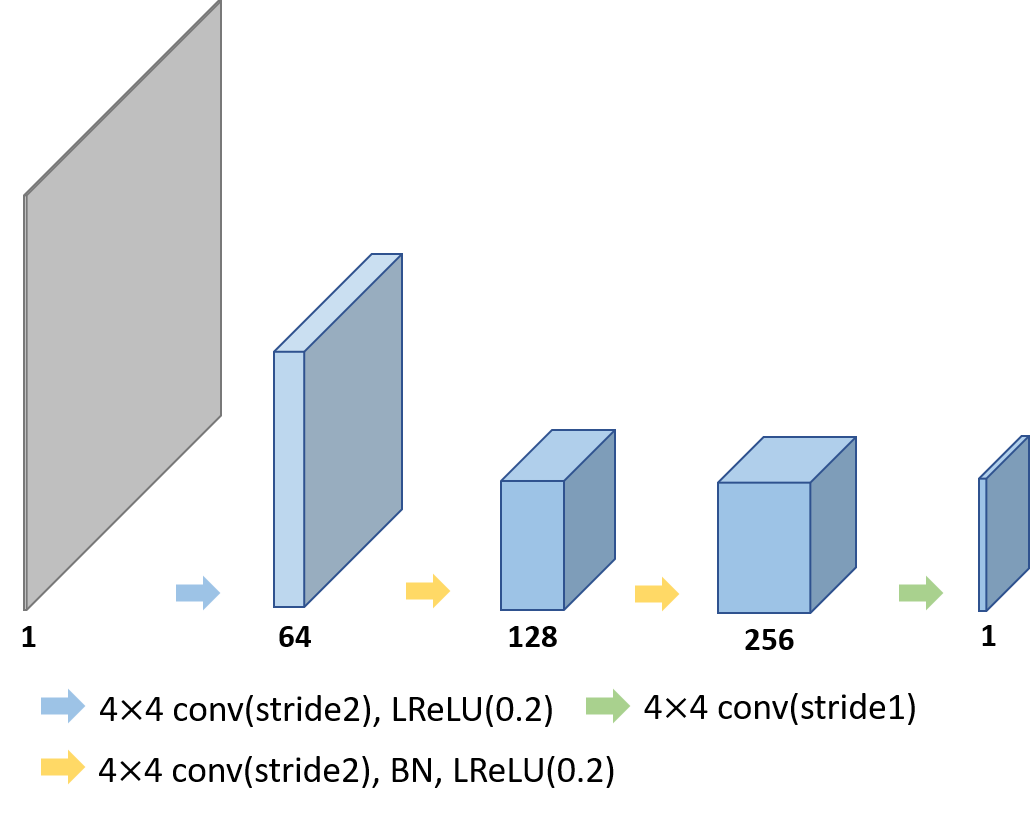}}
		\centerline{(b)}\medskip
	\end{minipage}
	\hspace{-0.2cm}
	\caption{Network architecture for (a) $\Fc_i(\cdot)$ in the coupling layer,  and (b) PatchGAN discriminator.}
	\label{fig:NN and D}
\end{figure}

\begin{figure*}[ht!]
    \centering
    \begin{minipage}[b]{0.57\linewidth}
		\centerline{\includegraphics[width=\linewidth]{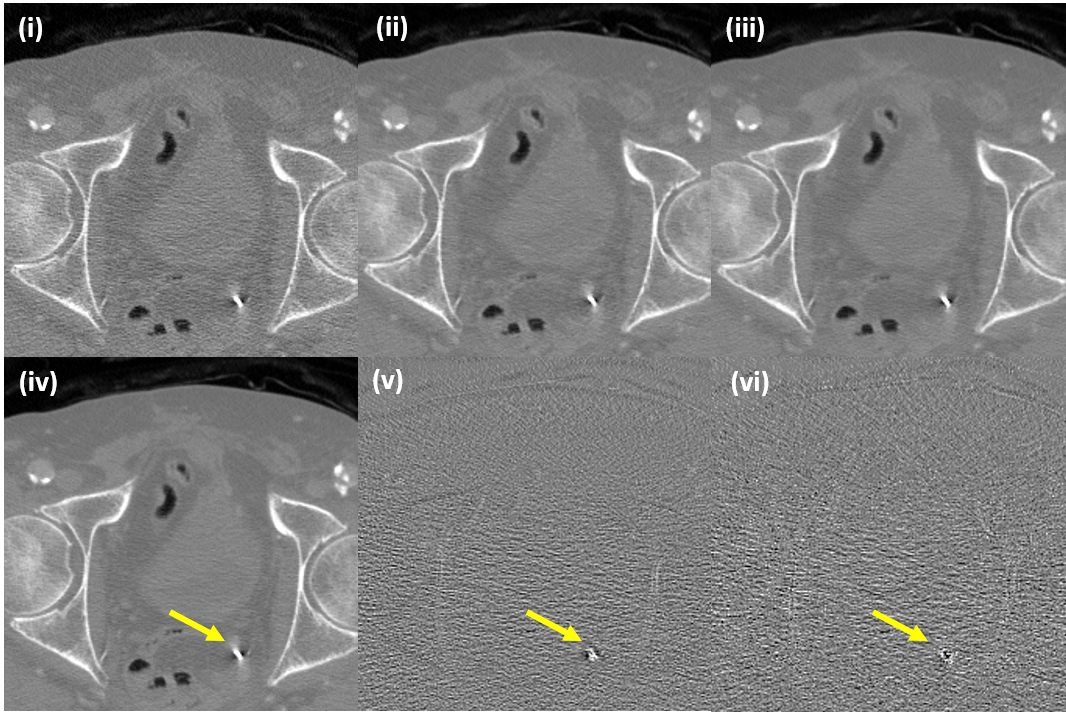}}
		\centerline{(a) Forward mapping}\medskip
	\end{minipage}
	\hspace{0.2cm}
	\begin{minipage}[b]{0.38\linewidth}
		\centerline{\includegraphics[width=\linewidth]{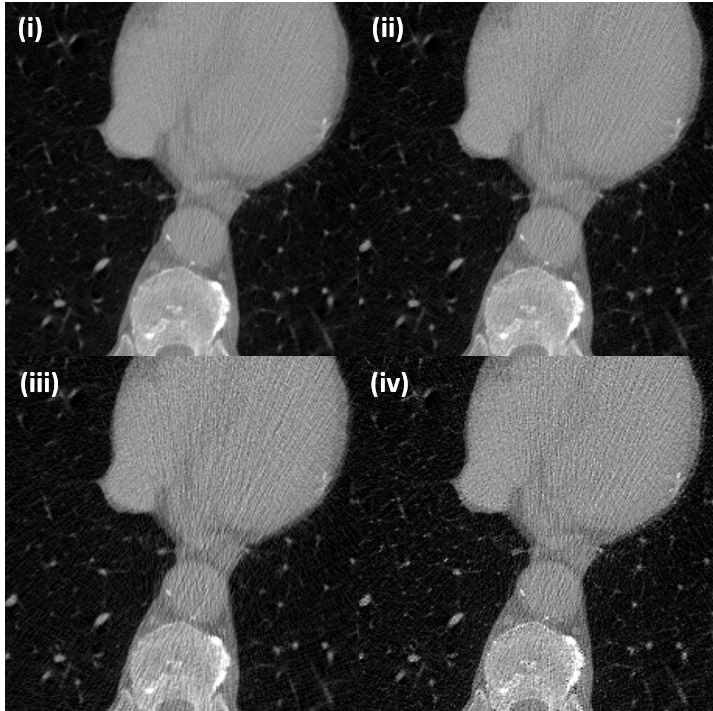}}
		\centerline{(b) Inverse mapping}\medskip
	\end{minipage}
	\caption{(a) Denoising results of AAPM data: (i) input low-dose image,  and denoising results from (ii) the AdaIN-based tunable CycleGAN, (iii) the proposed CycleGAN with an invertible generator, and (iv) target standard-dose image. (v)(vi) Difference images between (i) and the results (ii)(iii), respectively.
		(b) Synthetic noise generation  results:  (i) input standard-dose image, and synthetic low-dose image from (ii) the AdaIN-based tunable CycleGAN, (iii) the proposed CycleGAN with an invertible generator, and (iv) the target low-dose image.
		The intensity window of CT image is (-1000, 1000) [HU] and the intensity window of difference is (-200, 200) [HU].
		}
	\label{fig:AAPMresults}
\end{figure*}

\subsection{Implementation Details}
The invertible generator is constructed based on the flow of the invertible generator with $L=4$ as shown in Fig. \ref{fig:flowofgenerator}.
To extract the wavelet residual,  we use daub3 wavelets, and  
the level of wavelet decomposition was set to 6 for all datasets. 

{
The architecture of the neural network in the coupling layer (see \eqref{eq:forward}) is shown in Fig. \ref{fig:NN and D}(a). 
Basically, the architecture is composed of three convolution layers with spectral normalization\cite{miyato2018spectral, behrmann2021understanding}
followed by the multi-channel input single-channel channel output  convolution. 
The first and last convolution layer use 3$\times$3 kernel with stride of 1 and the second convolution layer uses 1$\times$1 kernel with stride of 1.
And the latent feature map channel size is 256.
Zero-padding is applied for the first and last convolution layer so that
at each stage, the height and width of the feature map are equal to the previous feature map.
}

The discriminator is constructed based on a PatchGAN architecture \cite{isola2017image}. The overall architecture of the discriminator is shown in Fig. \ref{fig:NN and D}(b), which is based on the PatchGAN discriminator composed of four convolution layers than five convolution layers. First two convolution layers use stride of 2, and the rest of the convolution layers use stride of 1. After the first and last convolution layer, we do not apply batch normalization. Except for the last convolution layer,  we apply LeakyReLU with a slope of 0.2 after the batch normalization. At the first convolution layer, which does not have batch normalization, LeakyReLU was applied after the convolution layer. The discriminator loss is calculated with the LSGAN loss \cite{mao2017least}.

%

For all datasets, {the network was trained with $\eta=10$ for $\ell_{\Yc}(G_{\thetab})$ in \ref{eq:Yc}}, using ADAM optimizer \cite{kingma2014adam} with $\beta_1 = 0.9$, $\beta_2 = 0.999$,  $\epsilon=1e^{-8}$,  and the mini-batch size of 1. The learning rate was initialized to $1\times 10^{-4}$ and halved in every 50,000 iterations. We trained network for 150,000 iterations on NVIDIA GeForce RTX 2080 Ti. Also, our code was implemented with Pytorch v1.6.0 and CUDA 10.1.

\begin{figure}[!hbt]
    \centerline{\includegraphics[width=\linewidth]{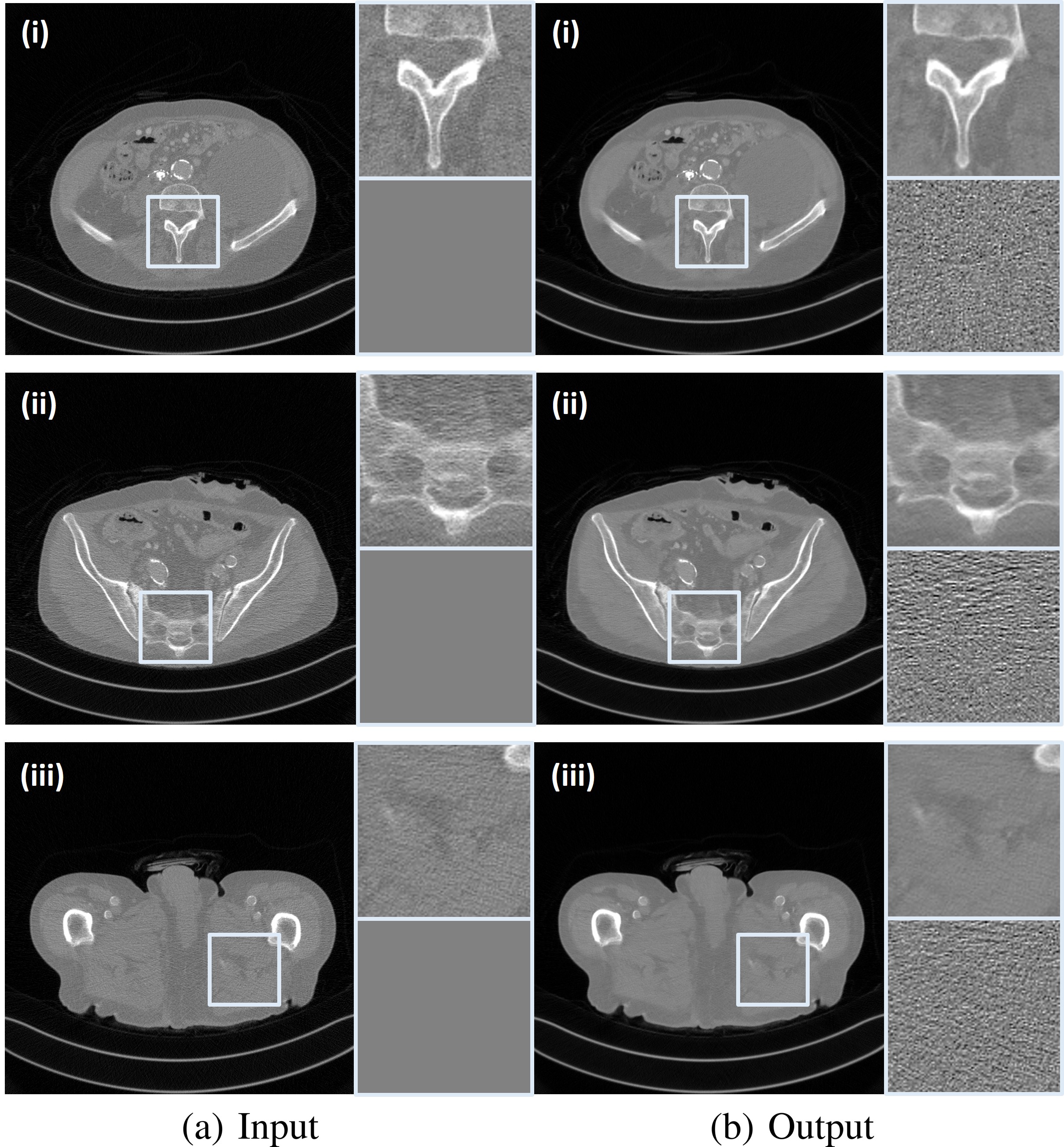}}
	\caption{{Denoising results for AAPM data using the proposed method. (a) Input low-dose CT images (i-iii), and (b) denoising results from the proposed cycle-free CycleGAN  (i-iii). 
	The intensity window of CT image is (-1000, 1000) [HU] and the intensity window of difference between input and output image is (-200, 200) [HU].
		}
		}
	\vspace*{-0.5cm}
	\label{fig:AAPMqualitatives}
\end{figure}

\begin{figure*}[!ht]
    \centering
    \begin{minipage}[b]{0.57\linewidth}
		\centerline{\includegraphics[width=\linewidth]{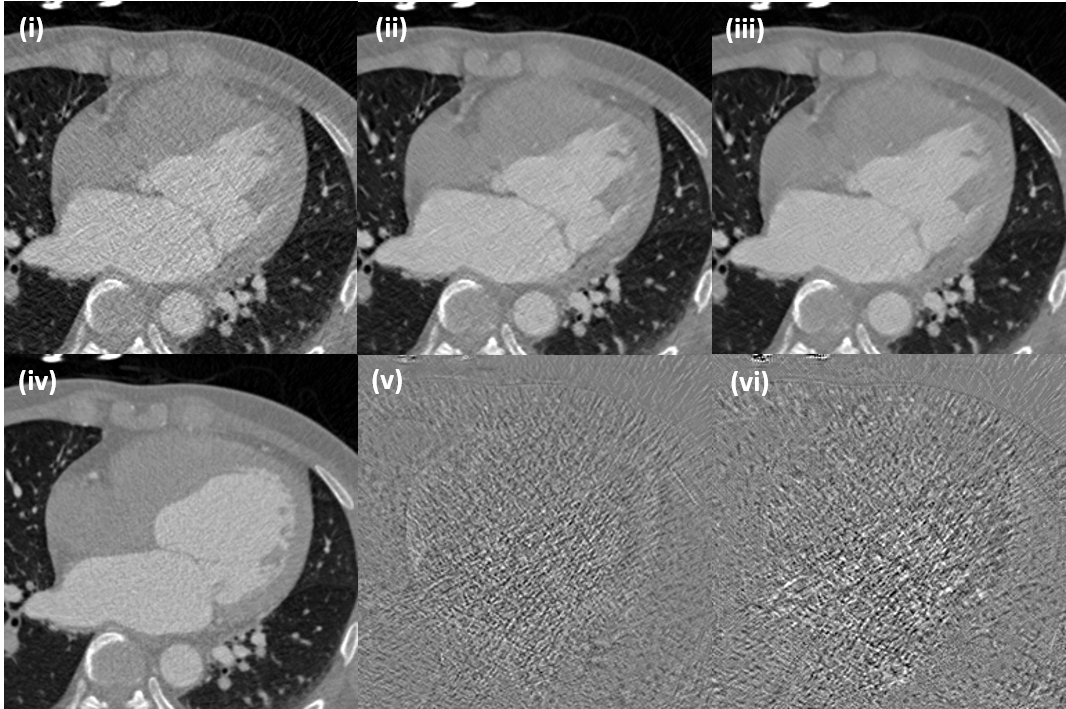}}
		\centerline{(a) Forward mapping}\medskip
	\end{minipage}
	\hspace{0.2cm}
	\begin{minipage}[b]{0.38\linewidth}
		\centerline{\includegraphics[width=\linewidth]{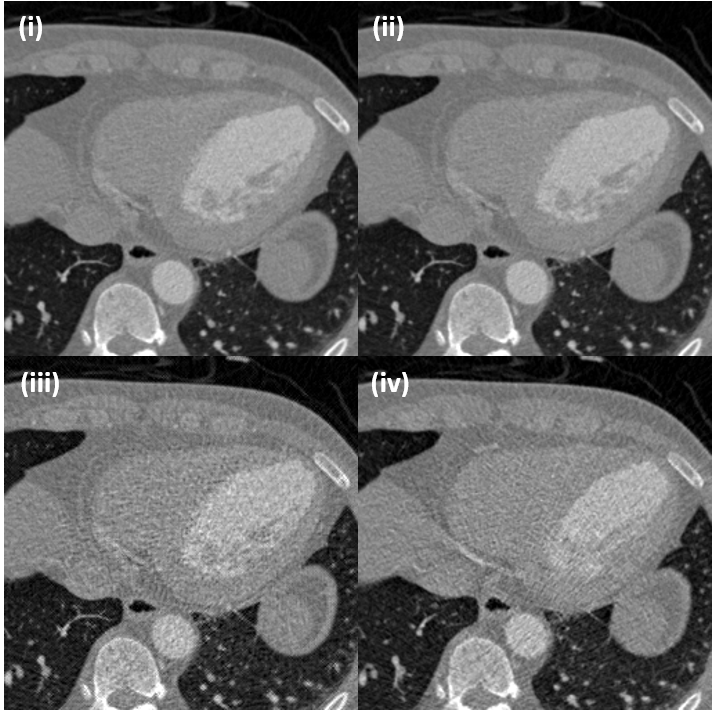}}
		\centerline{(b) Inverse mapping}\medskip
	\end{minipage}
	\caption{{(a) Denoising results of 20\% dose CT data. (i) Input low-dose image, denoising result from (ii) the AdaIN-based tunable CycleGAN, (iii) the proposed CycleGAN with an invertible generator, and (iv) target standard-dose image. (v-vi) Difference between (i) and the results (ii-iii).
	(b) Synthetic noise generation results:  (i) input standard-dose image, and synthetic low-dose image from (ii) the AdaIN-based tunable CycleGAN, (iii) the proposed CycleGAN with an invertible generator, and (iv) the target low-dose image. Note that the target is not perfectly aligned with the input, since
	there are no perfectly aligned high-dose images in in vivo experiments.
		The intensity window of CT image is (-1024, 1024) [HU] and the intensity window of difference is (-200, 200) [HU]
		}
		}
	\vspace*{-0.5cm}
	\label{fig:20doseresults}
\end{figure*}

\subsection{Quantitative Metrics}
For quantitative experiment analysis, we use the peak signal-to-noise ratio (PSNR) and the structural similarity index metric (SSIM) \cite{wang2004image}.
The PSNR is defined as follows:
\begin{eqnarray}
PSNR(x,y) = 20\log_{10}\frac{MAX_{\xb}}
{\lVert \xb-\yb \rVert}_{2} \ ,
\end{eqnarray}
where $\xb$ is the input image, $\yb$ is target image, and $MAX_{\xb}$ is possible maximum pixel value of image $\xb$.

The SSIM is defined as follows:
\begin{eqnarray}
SSIM(\xb,\yb) = \frac{(2\mu_{x}\mu_{y} + c_1)(2\sigma_{x,y} + c_2)}
{(\mu_{x}^2 + \mu_{y}^2 + c_1)(\sigma_{x}^2 + \sigma_{y}^2 + c_2)} \ ,
\end{eqnarray}
where $\mu$ is the average of image, $\sigma$ is the variance of image, $c_1 = (k_1L)^2$, $c_2 = (k_2L)^2$, $k_1 = 0.01$, $k_2 = 0.03$ as in the original paper\cite{wang2004image}.

\subsection{Comparative Methods}
We compared our method with the existing unsupervised LDCT denoising networks \cite{zhu2017unpaired, gu2021adain}.
For AAPM dataset, we compared our network performance with the conventional CycleGAN\cite{zhu2017unpaired} whose the generator is based on U-net\cite{ronneberger2015u} architecture. We also compared with AdaIN-based tunable CycleGAN \cite{gu2021adain},
which shows state-of-the-art performace for LDCT denoising.
For unpaired 20\%  dose CT scan datasets, we compare our method with AdaIN-based tunable CycleGAN.

For the training of the conventional CycleGAN, the images are cropped into $128\times128$ patches, the learning rate is initialized to $1\times 10^{-3}$, the network trained for 200 epochs, and the other training settings are set to the same as the proposed network training.
In AdaIN-based tunable CycleGAN, the same patch size was used, and the learning rate is initialized to $2\times 10^{-4}$, the network trained for 200 epochs, and the other training settings are the same as the proposed network training. For both comparative methods,  we used PatchGAN consisting of five convolution layers for discriminator architecture. 

\section{Experimental Results}\label{sec:result}
\subsection{AAPM CT dataset}

For the AAPM CT dataset, we first compare the noise reduction performance quantitatively with the conventional CycleGAN and AdaIN-based CycleGAN based on PSNR and SSIM.
As shown in Table \ref{tab:table3}, our network shows the highest PSNR and comparable SSIM values.

\begin{table}[h!]
  \begin{center}
    \caption{Quantitative results for AAPM dataset}
    \label{tab:table3}
    \resizebox{0.35\textwidth}{!}{
      \begin{tabular}{l|c|c}
        \toprule 
        \textbf{Network} & \textbf{PSNR} & \textbf{SSIM}\\
        \midrule 
        LDCT input & 30.468 & 0.695\\
        Conventional CycleGAN & 34.621 & 0.818\\
        AdaIN-based CycleGAN & 34.801 & \textbf{0.824}\\
        Proposed & \textbf{34.940} & 0.821\\
        \bottomrule 
        \end{tabular}
    }
  \end{center}
\end{table}

\begin{table*}[h!]
  \begin{center}
    \caption{Comparison of the network complexity in terms of trainable parameters}
    \label{tab:table4}
    \resizebox{0.65\textwidth}{!}{
      \begin{tabular}{c c|c c|c c}
        \toprule 
        \multicolumn{2}{c|}{Conventional CycleGAN}  &\multicolumn{2}{c|}{AdaIN CycleGAN\cite{gu2021adain}}  &  \multicolumn{2}{c}{Proposed} \\
        \midrule 
        Network & \# of Parameters & Network & \# of Parameters & Network & \# of Parameters\\
        \midrule 
        \midrule 
        $G_{\thetab}$ & 6,251,392 & $G_{\thetab}$ & 5,900,865 & $G_{\thetab}$ & 1,204,320\\
        $F_{\phib}$ & 6,251,392 & $F$ & 274,560 & - & -\\
        $D_{\xb}$ & 2,766,209 & $D_{\xb}$ & 2,766,209 & $D_{\xb}$ & 662,401\\
        $D_{\yb}$ & 2,766,209 & $D_{\yb}$ & 2,766,209 & - & - \\
        \midrule 
        Total & 18,035,202 & Total & 11,707,843 & Total & \textbf{1,866,721}\\
        \bottomrule 
        \end{tabular}
    }
  \end{center}
\end{table*}

\begin{table}[hbt!]
  \begin{center}
    \caption{Ablation study for invertible block components}
    \label{tab:table5}
    \resizebox{0.4\textwidth}{!}{
      \begin{tabular}{c c c c}
        \toprule 
        \multicolumn{2}{c}{Invertible block component}  &  \multirow{2}{*}{\textbf{PSNR}} & \multirow{2}{*}{\textbf{SSIM}}\\
        \cmidrule(){1-2}
        \textbf{Coupling layer} & \textbf{1$\times$1 conv} & & \\
        \midrule 
        \midrule 
        \cmark & \xmark & 29.931 & 0.691\\
        \xmark & \cmark & 30.706 & 0.704\\
        \cmark & \cmark & \textbf{34.940} & \textbf{0.821}\\
        \bottomrule 
        \end{tabular}
    }
  \end{center}
\end{table}


Fig. \ref{fig:AAPMresults}(a) shows representative
denoising results by various methods.
The resulting images are cropped at $ 256\times256 $ to more accurately visualize the denoising performance. The intensity of the CT images shown is (-1000, 1000) [HU] and the difference is (-200, 200) [HU].
Our CycleGAN with an invertible generator removes more noise components than the AdaIN-based CycleGAN method without losing any information. {As can be seen from the Fig.~\ref{fig:AAPMresults}(a)}, the proposed network (Fig.~\ref{fig:AAPMresults}(a-iii)) removes noise components around high-intensity metals more evenly than AdaIN-based CycleGAN(Fig.~\ref{fig:AAPMresults}(a-ii)).

To verify that the invertible generator can also properly perform the inverse mapping, we provide an inversely mapped output from  SDCT images {as shown in Fig.~\ref{fig:AAPMresults}(b).}
The resulting images are cropped at $256\times256$ resolution in order to more accurately visualize the improved quality. The intensity of the CT images shown is (-1000, 1000) [HU].
Even if the proposed method does not apply discriminator or loss for inverse mapping, the proposed method adds a reasonable level of noise to the SDCT, which makes the output of the LDCT appear closer than the AdaIN-based CycleGAN.


{In Fig.~\ref{fig:AAPMqualitatives}, there are three representative denoising results by the proposed CycleGAN with an invertible generator to verify the noise reduction performance qualitatively. 
The gray boxes  in the input low-dose CT images and the denoising result images are enlarged in order to more accurately visualize the noise reduction performance, and their difference from the input are also visualized.
The difference images clearly show the removed noise components. 
As can be seen from Fig.~\ref{fig:AAPMqualitatives}, the proposed method removes noise components evenly without any structural information loss. Therefore, it distinguishes bone and each soft tissue more clear.
}

\subsection{20\%  dose cardiac CT scan dataset}

The dataset does not have paired reference data so that quantitative comparison using PSNR and SSIM is not possible. 
Therefore, we qualitatively compared denoising performance. 
Intensity of the CT images is shown  (-1024, 1024) [HU], whereas the difference images are shown in (-200, 200) [HU] for 20\% dose.

{Fig. \ref{fig:20doseresults}(a) shows
denoising result by various methods.} 
Note that the target is not perfectly aligned with the input, since
	there are no perfectly aligned high-dose images in in vivo experiments.
Still, the visual inspection and the difference images from the input shows that 
our cycle-free CycleGAN with an invertible generator removes various noise components more uniformly than the AdaIN-based CycleGAN method
without incurring any structural distortion.
{In Fig. \ref{fig:20doseresults}(b),} SDCT images can be successfully converted to noisy images.
%
Even if the proposed method does not apply any discriminator or loss for inverse mapping, our method adds proper noise level to SDCT. 

{Also, Fig.~\ref{fig:20dosequalitatives} shows three representative  denoising results by the proposed cycle-free CycleGAN with an invertible generator. The gray boxes  in the low-dose inputs and denoised outputs are enlarged. The proposed method properly removes noise components from input low-dose CT images, so that each soft tissue in the resulting denoised images is distinguished clearly.
}

\section{Discussion}


As shown in Table \ref{tab:table4}, our network  uses only 10\% parameters of conventional CycleGAN, and 15\% parameters of AdaIN-based CycleGAN. 
This is because 
 we use a single generator and a single discriminator thanks to the invertibility. In addition, the networks in stable coupling layer are relatively light. Accordingly, the discriminator also requires relatively few parameters.
 
 Thanks to its efficient parameter requirement and memory consumption, cycle-free CycleGAN also reduced training time.
While training input image size with 256$\times$256 resolution, cycle-free CycleGAN shows training time of 12.2 iterations per second. However, our comparative model AdaIN-based tunable CycleGAN shows training time of 6.8 iterations per second. Accordingly, the training time is twice faster than that of AdaIN-based CycleGAN.

To investigate the optimality of our network architecture, we performed ablation studies.
In particular, we investigate the effect of invertible 1$\times$1 convolution and stable coupling layer, as these are the critical parts for the network design.
As shown in Table \ref{tab:table5}, both modules are critical. In particular, the results
without 1$\times$1 convolution layers show inferior performance than
using the invertible 1$\times$1 convolution.

\begin{figure}[!hbt]
    \centerline{\includegraphics[width=\linewidth]{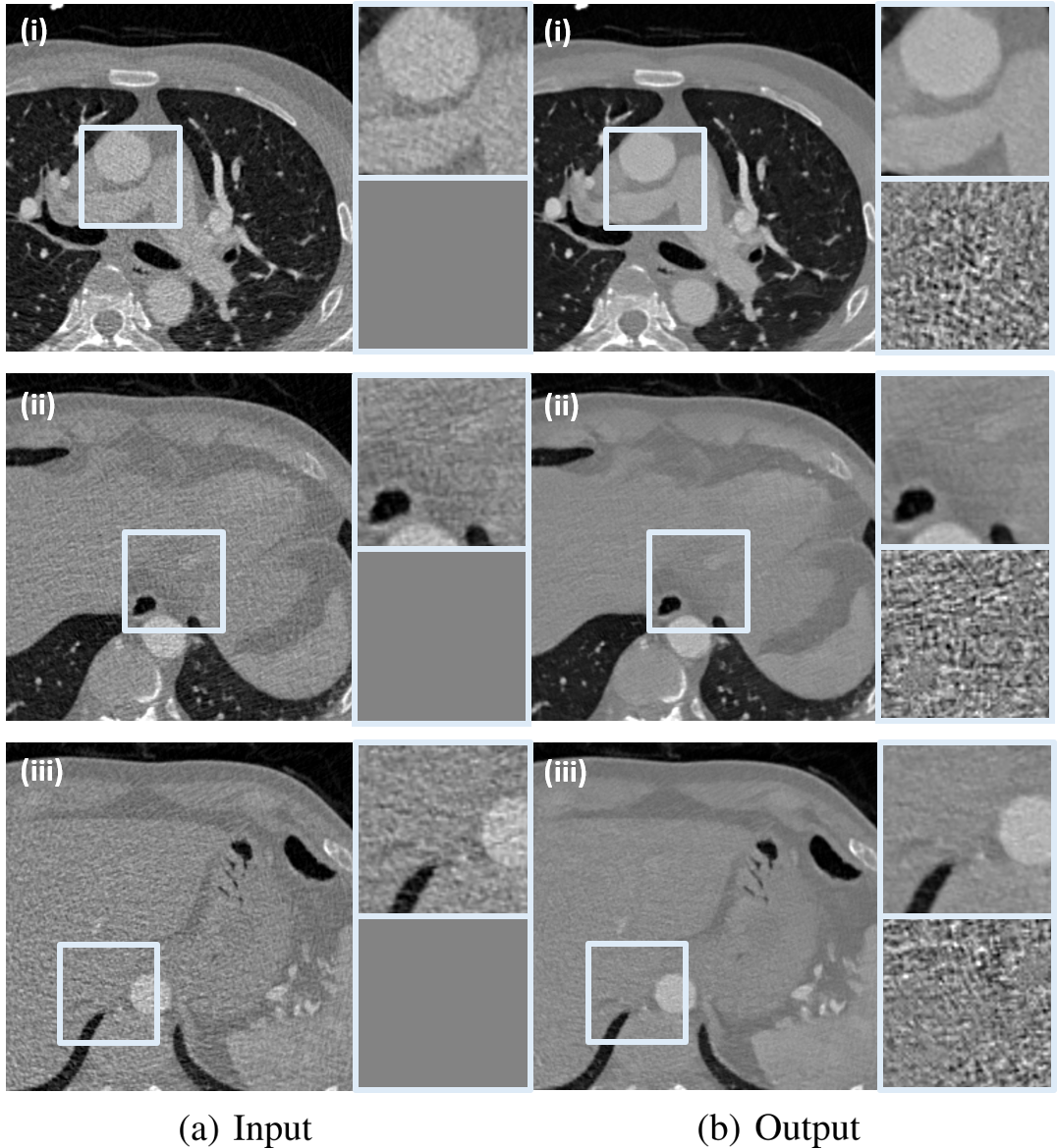}}
	\caption{{Denoising results for 20\% dose CT data using the proposed method. (a) Input low-dose CT images (i-iii), and (b) denoising results by the proposed cycle-free cycleGAN with an invertible generator (i-iii).
	The intensity window of CT image is (-1024, 1024) [HU] and the intensity window of difference is (-200, 200) [HU]
		}
		}
	\vspace*{-0.5cm}
	\label{fig:20dosequalitatives}
\end{figure}

\section{Conclusion}\label{sec:conclusion}

In this paper, we proposed a cycle-free CycleGAN architecture with an invertible generator.
Thanks to the invertibility, only a single pair of a generator and a discriminator is necessary, which
significantly reduced complexity.
 Although the number of trainable parameters are only 10\% of conventional CycleGAN and 15\% of AdaIN-based CycleGAN,
extensive experimental
 results confirmed that the proposed method shows better low-dose CT denoising performance with using significantly reduced learnable parameters.

\section{Acknowledgement}
This work was supported by the National Research Foundation (NRF) of Korea grant NRF-2020R1A2B5B03001980.
{The authors would like to thank Dr. Dong Hyun Yang from the University of Ulsan College of Medicine for providing the multiphase cardiac CT scan dataset.
The authors also thank the Mayo Clinic, the American Association of Physicists in Medicine (AAPM), and the National Institute of Biomedical Imaging and Bioengineering for providing the Low-Dose CT Grand Challenge dataset.}

\appendix

The derivation of the dual formula is simple modification of the technique in \cite{sim2020optimal}.
Consider the primal OT problem:
\begin{align*}
\inf\limits_{\pi \in \Pi(\mu,\nu)}\int_{\Xc\times \Yc}c(\xb,\yb;G_\thetab,F_\phib) d\pi(\xb,\yb) 
\end{align*}
where $\Pi(\mu,\nu)$ refers to the set of the joint distributions with the margins $\mu$ and $\nu$,
and the transportation cost is defined by
\begin{align*}
&c(\xb,\yb;G_\thetab,F_\phib)\notag\\
&= \|\xb-G_\thetab(\yb)\|+ \frac{1}{\beta}\|F_\phib(\xb)-\yb\| +\eta \|\yb-G_\thetab(\yb)\|
\end{align*}
We can easily show that
\begin{align*}
&\inf\limits_{\pi \in \Pi(\mu,\nu)}\int_{\Xc\times \Yc}c(\xb,\yb;G_\thetab,F_\phib) d\pi(\xb,\yb) \\
& = \inf\limits_{\pi \in \Pi(\mu,\nu)}\int_{\Xc\times \Yc}\tilde c(\xb,\yb) d\pi(\xb,\yb) + \eta\ell_\Yc(G_\thetab)
\end{align*} 
where $\ell_\Yc(G_\thetab)$ is defined in \eqref{eq:Yc} and
\begin{align*}
\tilde c(\xb,\yb)= \|\xb-G_\thetab(\yb)\|+ \frac{1}{\beta}\|F_\phib(\xb)-\yb\|
\end{align*}

We now define the optimal joint measure $\pi^*$ for the primal problem. 
Using the Kantorovich dual formulations, we have the following two equalities:
\begin{align}
K:=& \int_{\Xc\times \Yc} \tilde c(\xb,\yb) d\pi^*(\xb,\yb)\notag \\
=  \max_{\varphi} & \left\{ \int_\Yc \inf_\xb \{ \tilde c(\xb,\yb) -\varphi(\xb) \}d\nu(\yb)  + \int_\Xc \varphi(\xb) d\mu(\xb)  \right\} \label{eq:e1}\\ 
=\max_{\psi} &\left\{ \int_\Xc \inf_\yb \{\tilde c(\xb,\yb) -\psi(\yb) \}d\mu(\xb) + \int_\Yc \psi 
(\yb)d\nu(\yb)\right\} \label{eq:e2}
\end{align}
Using  $1$-Lipschitz continuity of the Kantorovich potential $\varphi$, we have
\begin{align*}
&-\varphi(G_\thetab(\yb))\leq \| G_\thetab(\yb) - \xb\| -\varphi(\xb) \leq  c(\xb,\yb)-\varphi(\xb) 
\end{align*}
Using  $1/\beta$-Lipschitz continuity of the Kantorovich potential $\psi$, we have
\begin{align*}
&-\psi(F_\phib(\xb))  \leq  \frac{1}{\beta}\|\yb-F_\phib(\xb)\| -\psi(\yb) \leq  c(\xb,\yb) -\psi(\yb) 
\end{align*}
This leads to two lower bounds and by taking the average of the two, we have
\begin{align*}
K\geq \frac{1}{2}\ell_{GAN}(G_\thetab,F_\phib;\varphi,\psi)
\end{align*}
where $\ell_{GAN}$  is defined in \eqref{eq:Disc}.
For and upper bound,  instead of finding the $\inf_\xb$, we choose {$\xb=G_\thetab(\yb)$} in \eqref{eq:e1}; similarly, instead of $\inf_\yb$,  we chose $\yb=F_\phib(\xb)$
in \eqref{eq:e2}. By taking the average of the two upper bounds, we have
\begin{eqnarray*}
K&\leq \frac{1}{2}\left\{\ell_{GAN}(G_\thetab,F_\phib;\varphi,\psi)+\ell_{cycle}(G_\thetab,F_\phib) \right\}
\end{eqnarray*}
where   $\ell_{cycle}$  is defined \eqref{eq:cycleloss}.
The remaining part of the proof for the dual formula is a simple repetition of the techniques  in \cite{sim2020optimal}.

\bibliographystyle{IEEEtran}
\bibliography{ref,biblio_book}

\end{document}